\documentclass{article}

\usepackage{PRIMEarxiv}

\usepackage{amssymb}
\usepackage{tikz}
\usepackage{amsthm, amsmath}
\usepackage{lscape}

\newtheorem{theorem}{Theorem}[section]
\newtheorem{proposition}[theorem]{Proposition}
\newtheorem{corollary}[theorem]{Corollary}
\newtheorem{lemma}[theorem]{Lemma}
\newtheorem{rmk}[theorem]{Remark}

\newtheorem{example}{Example}

\usepackage[utf8]{inputenc} 
\usepackage[T1]{fontenc}    
\usepackage{hyperref}       
\usepackage{url}            
\usepackage{booktabs}       
\usepackage{amsfonts}       
\usepackage{nicefrac}       
\usepackage{microtype}      
\usepackage{lipsum}
\usepackage{fancyhdr}       
\usepackage{graphicx}       
\graphicspath{{media/}}     

\newcommand{\bigzero}{\mbox{\normalfont\bfseries 0}}
\newcommand{\bigjota}{\mbox{\normalfont\bfseries J}}
\newcommand{\real}{{\rm I\!R}}

\title{On the characteristic polynomial of the $A_\alpha$-matrix for some operations of graphs
}

\author{
  João Domingos G. da Silva Jr. 
  Departamento de Engenharia de Produção\\
  Centro Federal de Educação Tecnológica do Rio de Janeiro \\
  Rio de Janeiro, Brazil\\
  \texttt{joao.dgomes@gmail.com} \\
  \And
  Carla Silva Oliveira 
  Departamento de Matemática \\
  Escola Nacional de Ci\^encias Estat\'{\i}sticas \\
  Rio de Janeiro, Brazil\\
  \texttt{carla.oliveira@ibge.gov.br} \\
  \And
  Liliana Manuela G. C. da Costa 
  Departamento de Matemática \\
  Col\'egio Pedro II \\
  Rio de Janeiro, Brazil\\
  \texttt{lmgccosta@gmail.com} \\
}

\begin{document}
\maketitle

\begin{abstract}
Let G be a graph of order $n$ with adjacency matrix $A(G)$ and diagonal matrix of degree $D(G)$. For every $\alpha \in [0,1]$, Nikiforov \cite{VN17} defined the matrix $A_\alpha(G) = \alpha D(G) + (1-\alpha)A(G)$. In this paper we present the $A_{\alpha}(G)$-characteristic polynomial when $G$  is obtained by coalescing two graphs, and if $G$ is a semi-regular bipartite graph we obtain the $A_{\alpha}$-characteristic polynomial of the line graph associated to $G$. Moreover, if $G$ is a regular graph we exhibit the $A_{\alpha}$-characteristic polynomial for the graphs obtained from some operations.
\end{abstract}


\section{Introduction}\label{intro}
 
Let $G=(V,E)$ be a simple graph such that $\vert V\vert = n$ and $\vert E \vert = m$. For each vertex $v \in V$ the degree of $v$, denoted by $d(v)$, is defined by the number of edges incident to $v$. The minimum degree of $G$, is denoted by $\delta(G) := \min \{ d(v); v \in V \} $ and the maximum degree of $G$ by $\Delta(G) := \max\{d(v); v \in V\}$. The graph $G$ is called $r$-regular if each vertex of $G$ has degree $r$. We denote the path with $n$ vertices by $P_n$, the complete graph by $K_n$ and the star by $K_{1,n-1}$. The complete bipartite graph with order $n = n_1 + n_2$ is denoted by, $K_{n_1,n_2}$. A graph $G$ is called semi-regular bipartite, with parameters $(n_1, n_2, r_1, r_2)$, if $G$ is bipartite such that $V = V_1 \cup V_2$ where $n_1 = \vert V_1 \vert$ and $n_2 = \vert V_2 \vert$, and the vertices in the same partition has the same degree, in other words, $n_1$ vertices has degree $r_1$ and $n_2$ vertices has  degree $r_2$, where $n_1r_1 = n_2r_2$. The complement of $G$, denoted by $\overline{G}= (\overline{V},\overline{E})$, is the graph  obtained from $G$ with the same vertex set, $\overline{V} = V$, and $v_iv_j \in \overline{E}$ if and only if $v_iv_j \notin E$.

Let  $G = (V,E)$ and $H = (W, F)$ be disjoint graphs that is $V \cap W = \varnothing $. The union of $G$ and $H$ is the graph, denoted by $G \cup H$, such that $G \cup H = (V \cup W, E \cup F)$. The coalescence of $G$ and $H$, denoted by $G \cdot H$, is a graph of order $\vert V \vert + \vert W \vert - 1$ that can be obtained from the graph $G \cup H$ by identifying some vertex of $G$ with some vertex of $H$. The line graph of $G$, denoted by $l(G)$, is obtained the following way: for each edge in $G$, make a vertex in $l(G)$ and for two edges in $G$ that have a vertex in common, make an edge between their corresponding vertices in $l(G)$. 

The subdivision of a graph $G$, $S(G)$, is the graph obtained by inserting a new vertex on each of the edges of $G$. The graph $R(G)$ is obtained from $G$ by adding, for each edge $uv \in E$, a new vertex whose neighbours are $u$ and $v$. We denote by $Q(G)$ the graph obtained from $G$ by inserting a new vertex into each edge of $G$, joining by edges those pairs of new vertices which lie on adjacent edges of $G$. The total graph of  $G$, denoted by $T(G)$, is the graph whose vertices are the vertices and edges of $G$, such that two vertices of $T(G)$ are adjacent if, and only if, the corresponding elements of $G$ are adjacent or incident. An example of each of these graphs is shown in Fig. \ref{fig::SRQT}.

\begin{figure}[!h]%
    \centering
    \includegraphics[width=0.65\textwidth]{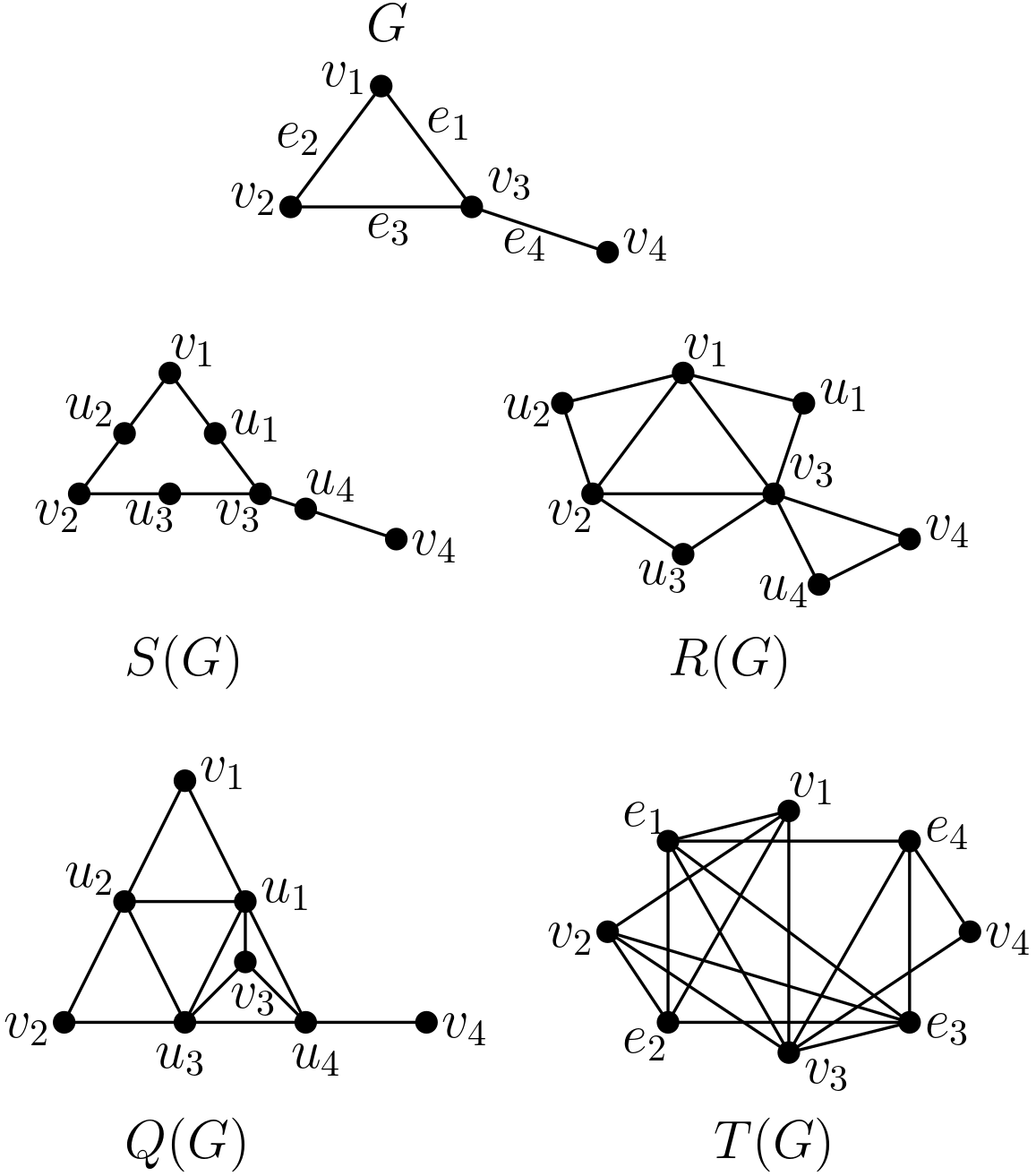}
    \vspace{0.5cm}
    \centering
    \caption{$G$, $S(G)$, $R(G)$, $Q(G)$ and $T(G)$ graphs.}
    \label{fig::SRQT}
\end{figure}

A double broom $B(q,n,m)$ is a graph obtained from the path $P_q$ by attaching $n$ pendants edges in one pendant vertex of $P_q$ and $m$  pendants edges in the other one. The pineapple graph, denoted by $K_m^n$, can be obtained by a coalescence of a vertex of degree $n$ of the star $K_{1,n}$ with a vertex of the complete graph $K_m$. A double star $D(m,n)$ is the graph obtained by a coalescence of $K_{1,m+1}$ and $K_{1,n}$ identifying one vertex of $K_{1,m+1}$ of degree one with the central vertex of $K_{1,n}$.

The adjacency matrix of $G$, denoted by $A = A(G) = [a_{ij}]$, is a square and symmetric matrix of order $n$, such that $a_{ij} = 1$ if $v_i$ is adjacent to $v_j$ and $a_{ij} = 0$ otherwise. The degree matrix of $G$, denoted by $D(G)$, is the diagonal matrix that has the degree  of the vertex $v_i$, $d(v_i)$, in the $i^{th}$ position. We denote the $n \times n$ all ones matrix by $\bigjota_n$ and the identity matrix by $I_n$.

Let $M$ be an $n \times n$ matrix. The $M$-characteristic polynomial is defined by 
\begin{equation*}
    P_M(\lambda) = \det(\lambda I_n - M) = \vert \lambda I_n - M \vert
\end{equation*}
and the roots of $P_M(\lambda)$ are called the $M$-eigenvalues. If $M$ is symmetric, the $M$-eigenvalues are real and we shall index them in non-increasing order, denoted by $\lambda_1(M) \geq \ldots \geq \lambda_n(M)$. The collection of $M$-eigenvalues together with their multiplicities is called the $M$-spectrum, denoted by $\sigma(M)$.

In $2017$ Nikiforov, \cite{VN17}, defined for any real $\alpha \in [0,1]$ the convex linear combination $A_\alpha(G)$ of $A(G)$ and $D(G)$ the following way: 
$$A_\alpha(G) = \alpha D(G) + (1-\alpha)A(G) , \ \ \alpha \in [0,1].$$

Among the various issues involving the $A_\alpha$- matrix we can point at: obtaining values for alpha for which the $A_\alpha$ matrix is positive semi-definite (\cite{VN17}, \cite{BRONDANI2020}, \cite{NIKIFOROV2017156}), obtaining the $A_\alpha$-spectrum for given families of graphs (\cite{VN17}, \cite{BRONDANI2019209}, \cite{MUHAMMAD2020}), determination the $A_\alpha$-spectrum of operations between graphs (\cite{VN17}, \cite{Li2019TheS}, \cite{CHEN2019343}, \cite{LIN2018210}), look for graphs determined by their $A_\alpha$-spectrum(\cite{Lin2017GraphsDB}, \cite{Tahir2018}), bounds for their eigenvalues and determinations of the extremal graphs (\cite{LIN2018210}, \cite{LIU2020111917}, \cite{LIN2018430}, \cite{WANG2020210}, \cite{LIU2020347}, \cite{SP2022}).

The determination of the characteristic polynomial of this matrix is also a subject of study and is present in several papers (\cite{CHEN2019343}, \cite{Lin2017GraphsDB}, \cite{Tahir2018},\cite{LIU2018274}). In this paper we show some results involving the characteristic polynomial of $A_\alpha$-matrix of a graph obtained by the coalescence of two graphs   and we make a study of the relationship between the characteristic polynomial of $A_\alpha(l(G))$-matrix with either the characteristic polynomial of $A_\alpha(G)$ and $A(G)$-matrix when $G$ is a regular graph, and with $A(G)$-matrix when $G$ is a semi-regular bipartite graph. Moreover, we present the characteristic polynomial of $A_\alpha$-matrix of graphs $S(G), R(G), Q(G)$ and $T(G)$ when $G$ is a regular graph.

This paper is organized as follows. In Section \ref{preliminaries} we introduce some definitions and present some preliminaries results. After, in Section \ref{results}  we show the main results obtained for the $A_\alpha$-characteristic polynomial for coalescence of two graphs, line graph and some graphs resulting from other operations when $G$ is a regular graph.


\section{Preliminaries} \label{preliminaries}

In this section we present some propositions, lemmas and theorems that will be useful to prove the main results of the paper. Let $P$ be a set and $ \pi = \{P_1,\ldots,P_m\}$ be its partition, and suppose that the matrix $M$ has its rows and its columns indexed by $P$.  The partition of $M$ with respect to $\pi$ is said an equitable partition if each sub-matrix $M_{ij}$, formed by the rows whose indices are in $P_i$ and the respective columns whose indices are in $P_j$, has a constant row sum equals to $q_{ij}$. Furthermore, the matrix $N = (q_{ij})_{1 \leq i,j \leq m}$ is called the quotient matrix of $M$ with respect to $\pi$.  Theorem \ref{EquitPart} relates the $M$-eigenvalues with the $N$-eigenvalues and Theorem \ref{theo::inv_block_matrix} presents how to compute the determinant of a block matrix.

\begin{theorem} \label{EquitPart}
	\cite{ADB_Equitable} Let $M$ be a square matrix of order $n$ and suppose that $M$ has a equitable partition $\pi = \{P_1,P_2,\ldots,P_k\}$. Let $N$ be the quotient matrix of $M$ with respect to the partition $\pi$. Then the eigenvalues of $N$ are eigenvalues of $M$.
\end{theorem}

\begin{theorem} \label{theo::inv_block_matrix}~\cite{horn2013matrix, silvester_2000} Let $M$ be a block matrix given by
	$$M = \begin{bmatrix}
		A & B \\
		C & D
	\end{bmatrix}.$$
	If $A$ and $D$ square matrix, then
	$$\det(M) = \begin{cases}
		\det(A) \det(D - CA^{-1}B), \mathrm{ \ \ if \ \ } A^{-1} \mathrm{\ \ exists \ \ }\\
		\det(D) \det(A - BD^{-1}C), \mathrm{ \ \ if \ \ } D^{-1} \mathrm{\ \ exists. \ \ }
	\end{cases}$$
\end{theorem}

Proposition \ref{prop::semi_regular_spectrum} exhibits two eigenvalues of $A(G)$, where $G$ is a semi-regular bipartite graph.

\begin{proposition} \label{prop::semi_regular_spectrum}
	\cite{cvetkovic2009introduction} Let $G$ be a semi-regular bipartite graph with integers parameters $n_1,n_2,r_1 \text{ and }r_2$, such that $n_1 \geq n_2$. Then $\pm \sqrt{r_1r_2}$ are eigenvalues of $A(G)$.
\end{proposition}

The incident matrix of $G$, denoted by $B = B(G) = [b_{ij}]$, is a matrix of order $n \times m$ such that $b_{ij} = 1$ if $e_j$ is an incident edge at $v_i$ and $b_{ij} = 0$ otherwise. 

\begin{lemma} \label{lemma::incident_BTB}
	\cite{cvetkovic2009introduction} Let $G$ be a graph with $m$ edges and $B = B(G)$ the incident matrix of $G$. Then $B^TB = 2I_m + A(l(G))$.
\end{lemma}
\begin{lemma} \label{lemma::incident_BBT}
	\cite{cvetkovic2009introduction} Let $G$ be a graph with $n$ vertices and $m$ edges. Consider $B$ and $D$ the incident and the degree matrix of $G$, respectively. Then $BB^T = D(G) + A(G)$.
\end{lemma}

Theorem \ref{theo::charpol_semi_regular} presents a relation between the $A(l(G))$-characteristic polynomial and $A(G)$-characteristic polynomial, when $G$ is a semi-regular bipartite graph.

\begin{theorem} \label{theo::charpol_semi_regular}
	\cite{cvetkovic2009introduction} Let $G$ be a semi-regular bipartite graph with $n_1$ independent vertices of degree $r_1$, $n_2$ independent vertices of degree $r_2$, with $ n_1 \geq n_2$. Then,
	\begin{equation*}
	    P_{A(l(G))}(\lambda) = (\lambda + 2)^\beta P_{A(G)}(\sqrt{\alpha_1\alpha_2})\sqrt{\left(\dfrac{\alpha_1}{\alpha_2}\right)^{n_1-n_2}},
	 \end{equation*}
	 where $\alpha_i= \lambda - r_i + 2$  for $i = 1,2$ and $\beta = n_1r_1 - n_1 - n_2$.
\end{theorem}

It is important to emphasize that Nikiforov exhibited the $A_\alpha$-spectrum of the complete graph, the complete bipartite graph and the star, as we will see in the following propositions.

\begin{proposition} \label{prop::complete_graph_spectrum}
\cite{VN17} The eigenvalues of $A_\alpha(K_n)$ are $\lambda_1(A_\alpha(K_n)) = n-1$ and $\lambda_k(A_\alpha(K_n)) = \alpha n -1 \text{  for } 2 \leq k \leq n$.
\end{proposition}

\begin{proposition} \label{prop::complete_bipartite_spectrum}
\cite{VN17} Let $a \geq b \geq 1$. If $\alpha \in [0,1]$, the eigenvalues of $A_\alpha(K_{a,b})$ are
\begin{align*}
    \lambda_1(A_\alpha(K_{a,b})) &= \dfrac{1}{2}\left(\alpha(a+b) + \sqrt{\alpha^2(a+b)^2 + 4ab(1-2\alpha)} \right),\\
    \lambda_{\min}(A_\alpha(K_{a,b})) &= \dfrac{1}{2}\left(\alpha(a+b) - \sqrt{\alpha^2(a+b)^2 + 4ab(1-2\alpha)} \right),\\
    \lambda_k(A_\alpha(K_{a,b})) &= \alpha a \text{ for } 1 < k \leq b,\\
    \lambda_k(A_\alpha(K_{a,b})) &= \alpha b \text{ for } b < k < a+b.
\end{align*}
\end{proposition}

\begin{proposition} \label{prop::star_spectrum}
\cite{VN17} The eigenvalues of $A_\alpha(K_{1,n-1}))$ are
\begin{align*}
    \lambda_1(A_\alpha(K_{1,n-1})) &= \dfrac{1}{2}\left(\alpha n + \sqrt{\alpha^2n^2 + 4(n-1)(1-2\alpha)} \right),\\
    \lambda_{n}(A_\alpha(K_{1,n-1})) &= \dfrac{1}{2}\left(\alpha n - \sqrt{\alpha^2n^2 + 4(n-1)(1-2\alpha)} \right),\\
    \lambda_k(A_\alpha(K_{1,n-1})) &= \alpha \text{ for } 1 < k < n.
\end{align*}
\end{proposition}

In \cite{Tahir2018} (Theorem 3), the authors present a relationship between the $A_{\alpha}(G)$-characteristic polynomials and $A_{\alpha}(\overline{G})$-characteristic polynomials for any graph. In particular, Proposition  \ref{prop::charpol_graph_complement} presents this relationship for regular graphs.
\begin{proposition} \label{prop::charpol_graph_complement}
\cite{Tahir2018} Let $G$ be a $r$-regular graph on $n$ vertices. Then
\begin{equation*}
    P_{A_\alpha(\overline{G})}(\lambda) = (-1)^n \dfrac{\lambda + r +1 -n}{\lambda +r + 1- n \alpha}P_{A_\alpha(G)}(n \alpha -1 -\lambda) 
\end{equation*}
If $r, \lambda_2(A_\alpha(G)), \ldots, \lambda_n(A_\alpha(G))$ are the eigenvalues of $A_\alpha(G)$, then the eigenvalues of $A_\alpha(\overline{G})$ are $ n-r -1, n \alpha -1 - \lambda_2(A_\alpha(G)), \ldots ,  n \alpha - 1 - \lambda_n(A_\alpha(G))$.
\end{proposition}

The principal sub-matrix of a matrix is a sub-matrix obtained by removing rows and columns with the same indices, \cite{Goldberg59}. Denote by $M_\alpha(G)$ the principal sub-matrix of $A_\alpha(G)$ obtained by removing one row and one column. Now, we present some results involving the $M_\alpha(G)$-eigenvalues when $G$ is a complete bipartite graph, a star and a complete graph. 

\begin{proposition} \label{prop::bipartite_submatrix_spectrum}
	Let $G \cong K_{a,b}$, $ a \leq b$, and $\alpha \in [0,1]$. Then $\sigma(M_\alpha(G))$ is
	\begin{itemize}
		\item [(i)] {\footnotesize $\displaystyle \Biggl\{a \alpha^{(n-1-a)}, b \alpha^{(n-2-b)}, \frac{1}{2}\left( \alpha (a+b) + \sqrt{\alpha^2 (a+b)^2 - 4b(\alpha^2 + (a-1)(2\alpha - 1))} \right), \\
		\frac{1}{2}\left( \alpha(a+b) - \sqrt{\alpha^2 (a+b)^2 - 4b(\alpha^2 + (a-1)(2\alpha - 1))} \right)\Biggr\}$}%
		 when the row and column removed from $A_\alpha(G)$ is referent to the partition of cardinality $a$.
		\item [(ii)] {\footnotesize $\displaystyle \Biggl\{ a\alpha^{(n-2-a)},b\alpha^{(n-1-b)}, \frac{1}{2}\left( \alpha(a+b) + \sqrt{\alpha^2 (a+b)^2 - 4a(\alpha^2 + (b-1)(2\alpha - 1))}  \right), \\
		\frac{1}{2}\left( \alpha(a+b) - \sqrt{\alpha^2 (a+b)^2 + 4a(\alpha^2 + (b-1)(2\alpha - 1))} \right)\Biggr\}$} when the row and column removed from $A_\alpha(G)$ is referent to the partition of cardinality $b$.
	\end{itemize}
\end{proposition}
\begin{proof}
	Let $G \cong K_{a,b}$ with $a+b = n$ and $a \leq b$. So
		$$A_\alpha(G) = 
		\begin{bmatrix}
			b\alpha I_a  & (1-\alpha)\bigjota_{a \times b} \\
			(1-\alpha)\bigjota_{b \times a} & a\alpha I_b
		\end{bmatrix}.
		$$
	Firstly, suppose that $M_\alpha(G)$ is the principal sub-matrix of $A_\alpha(G)$ obtained by removing a row and the respective column referring to the partition with cardinality $a$. Then, we have
		$$M_\alpha(G) = 
		\begin{bmatrix}
			b\alpha I_{a-1}  & (1-\alpha)\bigjota_{(a-1) \times b} \\
			(1-\alpha)\bigjota_{b \times (a-1)} & a\alpha I_b
		\end{bmatrix}.$$
		Applying equitable partitions on $M_\alpha(G)$ we get the quotient matrix $M'_\alpha(G)$ given by
		$$M'_\alpha(G) = 
		\begin{bmatrix}
			b\alpha & b(1-\alpha) \\
			(a-1)(1-\alpha) & a \alpha 
		\end{bmatrix}.$$
		Then,
		\begin{align*}
			P_{M'_\alpha(G)}(\lambda) = \vert \lambda I - M'_\alpha(G) \vert &= (\lambda - b\alpha)(\lambda - a \alpha) - b(a-1)(\alpha - 1)^2 \\
			&= \lambda^2 - (a+b)\alpha \lambda + ab(2\alpha - 1) + b(\alpha - 1)^2\\
			&= \lambda^2 - n\alpha \lambda + b\left( (\alpha - 1)^2(1-a) + \alpha^2 a \right)
		\end{align*}
		 which roots are
		\begin{align}
			\begin{split}
				&\displaystyle \frac{1}{2}\left( \alpha (a+b) + \sqrt{\alpha^2 (a+b)^2 - 4b(\alpha^2 + (a-1)(2\alpha - 1))} \right) \text{ and}\\ &\frac{1}{2}\left( \alpha(a+b) - \sqrt{\alpha^2 (a+b)^2 - 4b(\alpha^2 + (a-1)(2\alpha - 1))} \right)
			\end{split}
		\end{align}
	From Theorem \ref{EquitPart}, this roots are also roots of $P_{M_\alpha(G)}$.
		
	Let $e_i$ be the $i^{th}$ vector of the canonical basis of $\real^{n-1}$. Take the vectors $e_j - e_a, \ \ j = a+1, \ldots, n-1$. It is easy to see that $e_j - e_a$ are eigenvectors of the matrix $M_\alpha(G)$ associated to the eigenvalue $a\alpha$. So, the algebraic multiplicity of $a\alpha$ is at least $n-1-a$. Now, take the vectors $e_i - e_1,  \ \ i = 2, \ldots, a-1$. It is easy to see that $e_i - e_1$ are eigenvectors of the matrix $M_\alpha(G)$ associated to the eigenvalue $b\alpha$, whose algebraic multiplicity is at leat $n-2-b$.

    Now, suppose that $M_\alpha(G)$ is the sub-matrix of $A_\alpha(G)$ obtained by removing a row and the respective column referring to the partition $b$. Then we have
		$$M_\alpha(G) = 
		\begin{bmatrix}
			b\alpha I_a  & (1-\alpha)\bigjota_{a \times (b-1)} \\
			(1-\alpha)\bigjota_{(b-1) \times a} & a\alpha I_{(b-1)}
		\end{bmatrix}.$$
	Applying equitable partitions we obtain the quotient matrix $M'_\alpha(G)$ given by
		$$M'_\alpha(G) = 
		\begin{bmatrix}
			b\alpha & (1-\alpha)(b-1) \\
			a(1-\alpha) & a \alpha 
		\end{bmatrix}.$$
	Then,
		\begin{align*}
			P_{M'_\alpha(G)}(\lambda) =\vert\lambda I - M'_\alpha(G) \vert &=  (\lambda - a \alpha)(\lambda - b \alpha) - a(b-1)(\alpha - 1)^2\\
			&= \lambda^2 -(a+b)\alpha \lambda + ab(2\alpha-1) +a(\alpha - 1)^2\\
			&= \lambda^2 -n\alpha\lambda +a \left((\alpha - 1)^2(1-b) + \alpha^2b \right)
		\end{align*}
	which roots are
		\begin{align}
			\begin{split}
				&\displaystyle \frac{1}{2}\left( \alpha(a+b) + \sqrt{\alpha^2 (a+b)^2 - 4a(\alpha^2 + (b-1)(2\alpha - 1))}  \right) \text{ and }\\ &\frac{1}{2}\left( \alpha(a+b) - \sqrt{\alpha^2 (a+b)^2 - 4a(\alpha^2 + (b-1)(2\alpha - 1))} \right)
			\end{split}
		\end{align}

	Let $e_i$ be the $i^{th}$ vector of the canonical basis of $\real^{n-1}$. Take the vectors $e_j - e_{a+1}, \ \ j = a+2, \ldots,n-1$. We can see that $e_j - e_{a+1} $ are eigenvectors of the matrix $M_\alpha(G)$ associated to the eigenvalue $a\alpha $. So, the algebraic multiplicity of $a\alpha$ is at least $ n-2-a$. Consider the vectors $e_i- e_1, \ \ i = 2, \ldots,a $. It is easy to see that $e_i- e_1 $ are eigenvectors of the matrix $M_\alpha(G)$ associated to the eigenvalue $b\alpha $, whose algebraic multiplicity is at least $n-1-b$ and the result follows.
\end{proof}

As consequence of Proposition \ref{prop::bipartite_submatrix_spectrum} we obtain   the $\sigma(M_\alpha(K_{1,n-1}))$.
\begin{corollary} \label{cor::submatrix_star}
	Let $G \cong K_{1,n-1}$ and $\alpha \in [0,1]$. Then $\sigma(M_\alpha(G))$ is:
	\begin{enumerate}
	    \item[(i)] $\{\alpha^{(n-1)} \}$ when remove from $A_\alpha(G)$, the row and the column referent to the vertex of degree $n-1$.
	    \item [(ii)]
	    {\footnotesize
	    $\displaystyle \Biggl\{ \frac{1}{2} \left( \alpha n + \sqrt{\alpha^2(n^2 - 4) - 2(2\alpha-1)(n-2)} \right),$  $ \displaystyle \frac{1}{2} \left( \alpha n - \sqrt{\alpha^2(n^2 - 4) - 2(2\alpha-1)(n-2)} \right),\\
	    \alpha^{(n-3)} \Biggr\}$} when remove from $A_\alpha(G)$, the row and the column referent to the vertex of degree $1$.
	\end{enumerate}
\end{corollary}

\begin{proposition} \label{prop::CG_spec_submatrix}
	Let $G \cong K_n$ and $\alpha \in [0,1]$. Then, $\sigma(M_\alpha(G))$ is $\{ \alpha + n-2, n\alpha - 1^{(n-2)}\}.$
\end{proposition}
\begin{proof}
	Let $G \cong K_n$ and $\alpha \in [0,1]$. Suppose the matrix $M_\alpha(G)$ is the principal sub-matrix of $A_\alpha(G)$ obtained by removing a row and its respective column. So,
	 $$M_\alpha(G) = 
	\begin{bmatrix}
		(n-1)\alpha  & 1-\alpha & \dots & 1-\alpha \\
		1-\alpha & (n-1) \alpha&  \dots & 1-\alpha\\
		\vdots & \vdots & \ddots &\vdots \\
		1- \alpha & 1 \alpha & \ldots & (n-1) \alpha)
	\end{bmatrix}.$$

	Let $\mathbf{1}_{n-1}$ be the vector whose entries are all equal to $1$. It is easy to see that $\mathbf{1}_{n-1}$ is an eigenvector associated to the eigenvalue $\alpha + (n-2)$. Let $e_i$ be the $i^{th}$ vector of the canonical basis of $\real^{n-1}$. Take the vectors $e_j - e_1, \text{ for } j = 2, \ldots,n-1$. So, $e_j - e_1$ are eigenvectors of the matrix $M_\alpha(G)$ associated to the eigenvalue $n\alpha - 1 $, whose algebraic multiplicity is at least $ n-2$ and the result follows. 
\end{proof}

Let $G_u$ be the graph obtained from $G$ by adding a pendant edge at vertex $u \in V$ and let $G_{v_1,\ldots,v_s}$ be the graph obtained from $G$ by adding only one pendant edge at each of the vertices $v_1,\ldots,v_s \in V$. In Fig. \ref{fig::add_edge}, we show $G_{v_1,v_3,v_5,v_6}$ when $G \cong K_6$. We denote by $A_\alpha(G)_u$ the principal sub-matrix of $A_\alpha(G)$ obtained by removing the row and column corresponding to vertex $u$, whose characteristic polynomial is denoted by $P_{A_\alpha(G)_u}(\lambda)$. Proposition \ref{prop::charpol_many_edge} presents the characteristic polynomial of a graph obtained by adding a finite number of edges in the same vertex of the initial graph.

\begin{figure}[!h]%
    \centering
    \includegraphics[width=0.4\textwidth]{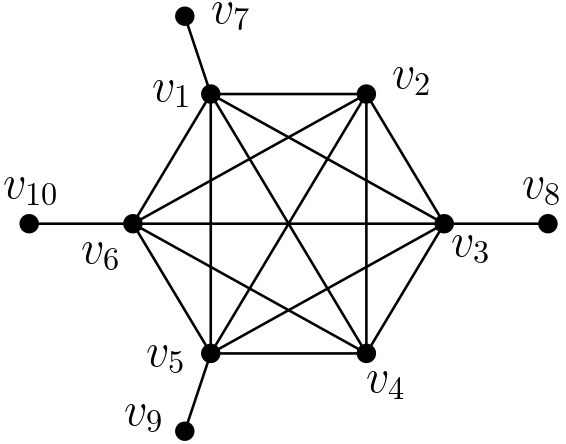}
    \vspace{0.5cm}
    \centering
    \caption{$G_{v_1,v_3,v_5,v_6}$ where $G \cong K_6$.}
    \label{fig::add_edge}
\end{figure}

\begin{proposition} \label{prop::charpol_many_edge}
\cite{CHEN2019343} Let $H$ be a connected graph with $\vert V(H) \vert \geq 2$, $v \in V(H)$ and $\alpha \in [0,1]$. Let $G$ be a graph obtained from $H$ by adding $s$ pendant edges to $v \in V(H)$. Then,
\begin{equation*}\label{eq::charpol_many_edge}
    P_{A_\alpha(G)}(\lambda) = (\lambda - \alpha)^sP_{A_\alpha(H)}(\lambda) - s(\alpha \lambda - 2 \alpha +1)(\lambda -\alpha)^{s-1}P_{A_\alpha(H)_v}(\lambda)
 \end{equation*}
\end{proposition}

From the Proposition \ref{prop::charpol_many_edge}, we obtain the Corollary \ref{cor2::charpol_many_edge}.

\begin{corollary}\label{cor2::charpol_many_edge}
    Let $G$ be a graph with $n$ vertices, $1 \leq s \leq n$ and $\alpha \in [0,1]$. Then 
    \begin{equation*} \label{eq3::charpol_many_edge}
        P_{A_\alpha(G_{v_1,\ldots,v_s})}(\lambda) = (\lambda - \alpha)^s P_{A_\alpha(G)}(\lambda) - (\alpha \lambda - 2 \alpha + 1) \sum_{j=1}^{s}(\lambda - \alpha)^{s-j}P_{A_\alpha(G)_{v_j}}(\lambda).
    \end{equation*}
\end{corollary}
\begin{proof}
We prove this result by induction on $s$. For $s= 1$ the result follows from Proposition \ref{prop::charpol_many_edge}. Suppose the result holds for $s - 1$, $s \geq 3$, that is,

\begin{equation} \label{eq1::charpol_many_edges}
    P_{A_\alpha(G_{v_1,\ldots,v_{s-1}})}(\lambda) = (\lambda - \alpha)^{s-1}P_{A_\alpha(G)}(\lambda) - (\alpha \lambda - 2 \alpha + 1) \sum_{j=1}^{s-1}(\lambda - \alpha)^{s-1 -j}P_{A_\alpha(G)_{v_j}}(\lambda).
\end{equation}

Consider $H \cong G_{v_1, \ldots, v_{s-1}}$. From Proposition \ref{prop::charpol_many_edge} we have,
\begin{equation} \label{eq4::charpol_many_edge}
		P_{A_\alpha(H_{v_s})}(\lambda) = (\lambda - \alpha)P_{A_\alpha(H)}(\lambda) - (\alpha \lambda - 2 \alpha + 1)P_{A_\alpha(G)_{v_s}}(\lambda).
	\end{equation}
Substituting (\ref{eq4::charpol_many_edge}) into (\ref{eq1::charpol_many_edges}) the result follows.
\end{proof}

\begin{example}
Let $G \cong K_n$, $1 \leq s \leq n$ and $\alpha \in [0,1]$. Consider the graph $G_{v_1,\ldots,v_s}$. It is easy to see that $P_{A_\alpha(G)_{v_i}}(\lambda) = P_{A_\alpha(G)_{v_j}}(\lambda), \ \ \forall 1 \leq i,j \leq n$. Doing $s = n$ and applying Corollary \ref{cor2::charpol_many_edge} together with Propositions \ref{prop::CG_spec_submatrix} and \ref{prop::complete_graph_spectrum} we have 
\begin{align*}
    \displaystyle P_{A_\alpha(G_{v_1,\ldots,v_n})} &= (\lambda - \alpha)^{n}(\lambda - n +1)(\lambda - \alpha n + 1)^{n-1} \\ 
    & - (\alpha \lambda - 2\alpha + 1) \left( \frac{(\lambda - \alpha)^{n-1} - 1}{\lambda - \alpha -1}\right)(\lambda - \alpha - n + 2)(\lambda - \alpha n + 1)^{n-2}\\
    & = (\lambda - \alpha n -1)^{n-2} \left[(\lambda - \alpha)^n g(\lambda) - \left( \frac{(\lambda - \alpha)^{n-1} - 1}{\lambda - \alpha -1}\right)f(\lambda)  \right]
\end{align*}

where $g(\lambda) = (\lambda -n + 1)(\lambda - \alpha n + 1)$ and $f(\lambda) = (\alpha \lambda - 2\alpha + 1)(\lambda - \alpha -n + 1)$
\end{example}

\begin{example}
 Let $G \cong K_{1,n-1}, \ \ 1 \leq s \leq n-1$ and $\alpha \in [0,1]$. Consider the graph $G_{v_1,\ldots,v_s}$ obtained from $G$ by adding only one pendant edge to each vertex of degree $1$. It is easy to see that if $d(v_i)= d(v_j) = 1 \ \ \forall 1 \leq i,j \leq n-1$ we have  $P_{A_\alpha(G)_{v_i}}(\lambda) = P_{A_\alpha(G)_{v_j}}(\lambda)$.  Doing $s = n-1$, the graph $G_{v_1,\ldots,v_s}$ is a starlike. Applying Corollaries \ref{cor::submatrix_star} and \ref{cor2::charpol_many_edge} together with Propositions \ref{prop::star_spectrum} we have
 
\begin{align*}
    \displaystyle P_{A_\alpha(G_{v_1, \ldots, v_s})} &= (\lambda - \alpha)^{n-1}(\lambda - c_1)(\lambda - c_2)(\lambda - \alpha)^{n-2} \\
    & \quad - (\alpha \lambda - 2 \alpha + 1)\left( \frac{(\lambda - \alpha)^{n-1} - 1}{\lambda - \alpha -1}\right)(\lambda - \alpha)^{n-3}(\lambda - c_3)(\lambda - c_4)\\
    & = (\lambda - \alpha)^{n-3} \Biggl[ (\lambda - \alpha)^n(\lambda - c_1)(\lambda - c_2) \\
    & \qquad - (\alpha \lambda - 2 \alpha + 1)\left( \frac{(\lambda - \alpha)^{n-1} - 1}{\lambda - \alpha -1}\right)(\lambda - c_3)(\lambda - c_4)\Biggr]\\
    &= (\lambda - \alpha)^{n-3} \Biggl[ (\lambda - \alpha)^n f(\lambda) + \left( \frac{(\lambda - \alpha)^{n-1} - 1}{\lambda - \alpha -1}\right) g(\lambda)\Biggr],
\end{align*}
where $f(\lambda) = \lambda^2 -\alpha n \lambda +(n-1)(2\alpha - 1)$, $g(\lambda) = (-\alpha \lambda +2\alpha -1)(\lambda^2 -\alpha n \lambda + (n-2)(2\alpha - 1) + \alpha^2)$, $c_1 = \dfrac{1}{2}\left( \alpha n + \sqrt{\alpha^2 n^2 + 4(n-1)(1-2\alpha)}\right)$, $c_2 = \dfrac{1}{2} \left( \alpha n - \sqrt{\alpha^2 n^2 + 4(n-1)(1-2\alpha)}\right)$, $c_3 = \dfrac{1}{2} \left( \alpha n + \sqrt{\alpha^2(n^2 - 4) - 2(2\alpha-1)(n-2)} \right)$ and $c_4 = \dfrac{1}{2} \left( \alpha n - \sqrt{\alpha^2(n^2 - 4) - 2(2\alpha-1)(n-2)} \right).$
\end{example}


\section{Main Results} \label{results}

This section presents the main results of this paper which involve $A_\alpha$-characteristic polynomial for coalescence of two graphs, the line graph, and the graphs  $S(G), R(G)$, $Q(G)$ and $T(G)$, when $G$ is a regular graph.

\begin{theorem} \label{theo::charpol_coalescence}
	Let $G \cdot H$ be a graph obtained from the coalescence of $G$ and $H$ identifying the vertices $u \in V(G)$ and $v \in V(H)$. Then,
\begin{align*} \label{eq::coalescence}
    \begin{split}
        P_{A_\alpha(G \cdot H)}(\lambda) &= P_{A_\alpha(G)}(\lambda)P_{A_\alpha(H)_v}(\lambda) + P_{A_\alpha(G)_u}(\lambda)P_{A_\alpha(H)}(\lambda) \\  
        &\qquad - \lambda P_{A_\alpha(G)_u}(\lambda)P_{A_\alpha(H)_v}(\lambda)
    \end{split}
\end{align*}
\end{theorem}

\begin{proof}
We know that
{\footnotesize
\begin{equation*} 
P_{A_\alpha(G \cdot H)}(\lambda) = \vert \lambda I  - A_\alpha(G \cdot H) \vert = 
	\begin{vmatrix}
		\lambda - (d_G(u) + d_H(v))\alpha & \nu^T & \rho^T \\
		\nu & \lambda I - A_\alpha(G)_u & 0 \\
		\rho & 0 & \lambda I - A_\alpha(H)_v
	\end{vmatrix},
\end{equation*}
}%
where $\nu$ and $\rho$ are vectors whose non-zero coordinates are $\alpha - 1$ for all vertices adjacent to $u$ in $G$ and  $v$ in $H$, respectively. As the determinant is a function multi-linear over the columns, we have

{\footnotesize 
\begin{align}
	\vert \lambda I  - A_\alpha(G \cdot H) \vert &= 
		\begin{vmatrix}
			\lambda - (d_G(u) + d_H(v))\alpha & \nu^T & \rho^T \\
			\nu & 0 & 0 \\
			\rho & 0 & \lambda I - A_\alpha(H)_v
		\end{vmatrix}\nonumber \\
		& \qquad +\begin{vmatrix}
			\lambda - (d_G(u) + d_H(v))\alpha & 0 & \rho^T \\
			\nu & \lambda I - A_\alpha(G)_u & 0 \\
			\rho & 0 & \lambda I - A_\alpha(H)_v
		\end{vmatrix}\nonumber\\
		&=\begin{vmatrix}
				\lambda - (d_G(u) + d_H(v))\alpha & \nu^T & \rho^T \\
				\nu & 0 & 0 \\
				\rho & 0 & 0
		\end{vmatrix}  + \begin{vmatrix}
		\lambda - (d_G(u) + d_H(v))\alpha & \nu^T & 0 \\
		\nu & 0 & 0 \\
		\rho & 0 & \lambda I - A_\alpha(H)_v
		\end{vmatrix}\nonumber \\
		& \qquad + \begin{vmatrix}
			\lambda - (d_G(u) + d_H(v))\alpha & 0 & \rho^T \\
			\nu & \lambda I - A_\alpha(G)_u & 0 \\
			\rho & 0 & \lambda I - A_\alpha(H)_v
		\end{vmatrix}\nonumber\\
	\begin{split} \label{eq1::charpol_coalescence}
		 &=\begin{vmatrix}
			\lambda - (d_G(u) + d_H(v))\alpha & \nu^T & 0 \\
			\nu & 0 & 0 \\
			\rho & 0 & \lambda I - A_\alpha(H)_v
		\end{vmatrix}\\
		& \qquad + \begin{vmatrix}
			\lambda - (d_G(u) + d_H(v))\alpha & 0 & \rho^T \\
			\nu & \lambda I - A_\alpha(G)_u & 0 \\
			\rho & 0 & \lambda I - A_\alpha(H)_v
		\end{vmatrix}
	\end{split}
\end{align}
}%
Adding and subtracting ${\scriptsize \begin{vmatrix}
		\lambda - (d_G(u) + d_H(v))\alpha & 0 & 0 \\
		\nu & \lambda I - A_\alpha(G)_u & 0 \\
		\rho & 0 & \lambda I - A_\alpha(H)_v
	\end{vmatrix}}$ into (\ref{eq1::charpol_coalescence}) and using, again, that the  determinant is a function multi-linear, we have
\begin{align}
	\begin{split} \label{eq2::charpol_coalescence}
		\vert \lambda I - A_\alpha(G \cdot H) \vert &= \begin{vmatrix}
				\lambda - (d_G(u) + d_H(v))\alpha & 0 & \rho^T \\
				\nu & \lambda I - A_\alpha(G)_u & 0 \\
				\rho & 0 & \lambda I - A_\alpha(H)_v
			\end{vmatrix}\\
		& \qquad + 
		\begin{vmatrix}
				\lambda - (d_G(u) + d_H(v))\alpha & \nu^T & 0 \\
				\nu & \lambda I - A_\alpha(G)_u & 0 \\
				\rho & 0 & \lambda I - A_\alpha(H)_v
			\end{vmatrix}\\
		& \qquad -
		\begin{vmatrix}
			\lambda - (d_G(u) + d_H(v))\alpha & 0 & 0 \\
			\nu & \lambda I - A_\alpha(G)_u & 0 \\
			\rho & 0 & \lambda I - A_\alpha(H)_v
		\end{vmatrix}
	\end{split}
\end{align}
Moreover, we know that:
{\footnotesize
\begin{flalign}
       & \begin{vmatrix}
		\lambda - (d_G(u) + d_H(v))\alpha & 0 & \rho^T \\
		\nu & \lambda I - A_\alpha(G)_u & 0 \\
		\rho & 0 & \lambda I - A_\alpha(H)_v
		\end{vmatrix} \nonumber \\
		& = \vert \lambda I - A_\alpha(G)_u \vert \begin{vmatrix}
			\lambda - (d_G(u) + d_H(v))\alpha & \rho^T \\
			\rho &  \lambda I - A_\alpha(H)_v
		\end{vmatrix} \nonumber\\
		& = \vert \lambda I - A_\alpha(G)_u \vert \begin{vmatrix}
			\lambda -  d_H(v)\alpha & \rho^T \\
			\rho &  \lambda I - A_\alpha(H)_v
		\end{vmatrix} + 
		\vert  \lambda I - A_\alpha(G)_u \vert 
		\begin{vmatrix}
			-d_G(u)\alpha & \rho^T \\
			0 &  \lambda I - A_\alpha(H)_v
		\end{vmatrix}  \nonumber \\
		&=P_{A_\alpha(G)_u}(\lambda)P_{A_\alpha(H)}(\lambda) - \alpha d_G(u) P_{A_\alpha(G)_u}(\lambda)P_{A_\alpha(H)_v}(\lambda), \label{eq3::charpol_coalescence}
\end{flalign}
}%
{\footnotesize
\begin{flalign}
       &\begin{vmatrix}
			\lambda - (d_G(u) + d_H(v))\alpha & \nu^T & 0 \\
			\nu & \lambda I - A_\alpha(G)_u & 0 \\
			\rho & 0 & \lambda I - A_\alpha(H)_v
		  \end{vmatrix} \nonumber\\
	    &= \vert \lambda I - A_\alpha(H)_v \vert  
	    	\begin{vmatrix}
	    	\lambda - (d_G(u) + d_H(v))\alpha & \nu^T \\
	    	\nu &  \lambda I - A_\alpha(G)_u
	    	\end{vmatrix} \nonumber \\
    	&=\vert  \lambda I - A_\alpha(H)_v \vert 
    		\begin{vmatrix}
    		\lambda -  d_G(u)\alpha & \nu^T \\
    		\nu &  \lambda I - A_\alpha(G)_u
    		\end{vmatrix} + 
    		\vert  \lambda I - A_\alpha(H)_v \vert 
    		\begin{vmatrix}
    		-d_H(v)\alpha & \nu^T \\
    		0 &  \lambda I - A_\alpha(G)_u
    		\end{vmatrix} \nonumber \\
    	&=P_{A_\alpha(H)_v}(\lambda)P_{A_\alpha(G)}(\lambda) - \alpha d_H(v) P_{A_\alpha(G)_u}(\lambda)P_{A_\alpha(H)_v}(\lambda), \label{eq4::charpol_coalescence}
\end{flalign}
}%
and
{\footnotesize
\begin{flalign}
      & \begin{vmatrix}
			\lambda - (d_G(u) + d_H(v))\alpha & 0 & 0 \\
			\nu & \lambda I - A_\alpha(G)_u & 0 \\
			\rho & 0 & \lambda I - A_\alpha(H)_v
			\end{vmatrix}&& \nonumber \\
			& = (\lambda -\alpha d_G(u) - \alpha d_H(v))P_{A_\alpha(G)_u}(\lambda)P_{A_\alpha(H)_v}(\lambda) \label{eq5::charpol_coalescence}
\end{flalign}
}%
Substituting (\ref{eq3::charpol_coalescence}), (\ref{eq4::charpol_coalescence}) and (\ref{eq5::charpol_coalescence}) in (\ref{eq2::charpol_coalescence}) we have
\begin{align*}
	P_{A_\alpha(G \cdot H)}(\lambda) &= P_{A_\alpha(G)_u}(\lambda)P_{A_\alpha(H)}(\lambda) - \alpha d_G(u) P_{A_\alpha(G)_u}(\lambda)P_{A_\alpha(H)_v}(\lambda)\\
	&+ P_{A_\alpha(H)_v}(\lambda)P_{A_\alpha(G)}(\lambda) - \alpha d_H(v) P_{A_\alpha(H)_v}(\lambda)P_{A_\alpha(G)_u}(\lambda)\\
	&- \lambda P_{A_\alpha(G)_u}(\lambda)P_{A_\alpha(H)_v}(\lambda) + \alpha d_G(u) P_{A_\alpha(G)_u}(\lambda)P_{A_\alpha(H)_v}(\lambda)\\
	&+ \alpha d_H(v) P_{A_\alpha(G)_u}(\lambda)P_{A_\alpha(H)_v}(\lambda)
\end{align*}
and consequently,
\begin{align*}
    P_{A_\alpha(G \cdot H)}(\lambda) &= P_{A_\alpha(G)_u}(\lambda)P_{A_\alpha(H)}(\lambda) + P_{A_\alpha(H)_v}(\lambda)P_{A_\alpha(G)}(\lambda) \\
    & \qquad - \lambda P_{A_\alpha(G)_u}(\lambda)P_{A_\alpha(H)_v}(\lambda).
\end{align*}
\end{proof}

In $2020$ Brondani et al, \cite{BRONDANI2020}, explained $P_{A_\alpha(K_m^n)}(\lambda)$. In Example \ref{ex::pineapple} we present this polynomial using the results presented in this paper.

\begin{example} \label{ex::pineapple}
Consider the pineapple $K_m^n$. Applying Corollary \ref{cor::submatrix_star}, Propositions \ref{prop::complete_graph_spectrum},  \ref{prop::star_spectrum} and \ref{prop::CG_spec_submatrix} in Theorem (\ref{theo::charpol_coalescence}), we have 
{\small
\begin{align*}
    P_{A_\alpha(K_m^n)}(\lambda) &= (\lambda -m +1)(\lambda - \alpha m +1)^{m-1}(\lambda - \alpha)^{n} \\
    & \qquad + (\lambda - \alpha -m+2)(\lambda - m\alpha +1)^{m-2}(\lambda - c_1)(\lambda - c_2)(\lambda - \alpha)^{n-1}\\
    & \qquad -\lambda(\lambda - \alpha - m + 2)(\lambda - m\alpha + 1)^{m-2}(\lambda - \alpha)^{n}\\
    & = (\lambda - \alpha)^{n-1}(\lambda - m\alpha +1)^{m-2} \Biggl[ (\lambda - m +1)(\lambda -m\alpha +1)(\lambda - \alpha)\\
    & \qquad + (\lambda - \alpha -m +2)(\lambda - c_1)(\lambda - c_2) - \lambda(\lambda -\alpha - m + 2)(\lambda - \alpha) \Biggr] \\
    & = (\lambda - \alpha)^{n-1}(\lambda - m\alpha +1)^{m-2} f(\lambda)
\end{align*}
where $f(\lambda) = (\lambda - m +1)(\lambda - m\alpha + 1)(\lambda - \alpha) - n(\lambda - \alpha -m +2)(\lambda \alpha - 2\alpha + 1)$, $c_1 = \dfrac{1}{2}\left( \alpha (n+1) + \sqrt{\alpha^2 (n+1)^2 + 4n(1-2\alpha)}\right)$ and $c_2 = \dfrac{1}{2} \left( \alpha (n+1) - \sqrt{\alpha^2 (n+1)^2 + 4n(1-2\alpha)}\right),$
}%
\end{example}

Let $G \cong K_{1,n}$ and $H \cong K_{1,m+1}$. Consider $G \cdot H$ be a graph obtained from the coalescence identifying one vertex of degree $1$ of $H$ with the vertex of degree $n$ of $G$. Fig. \ref{fig::double_star} shows $G \cdot H \cong D(m,n)$ with $m \neq n$.
\begin{figure}[h]
    \centering
    \includegraphics[width=0.4\textwidth]{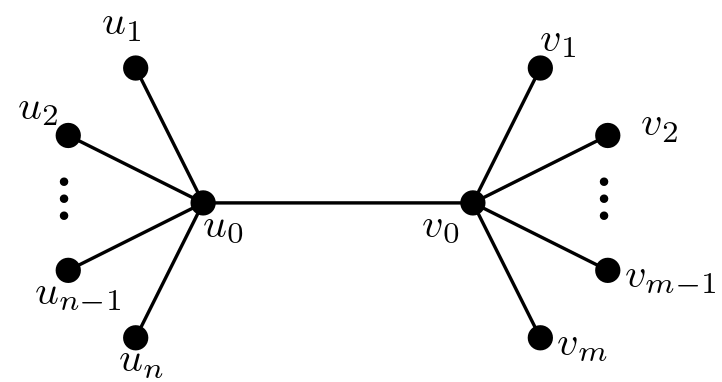}
     \vspace{0.5cm}
    \centering
    \caption{$D(m,n)$.}
    \label{fig::double_star}
\end{figure}
\begin{example} \label{ex::double_star}
Consider the double star $D(m,n)$.
Applying Corollary \ref{cor::submatrix_star}, Proposition \ref{prop::star_spectrum} and Theorem \ref{theo::charpol_coalescence} we have

{\footnotesize
\begin{align*}
    P_{A_\alpha(D(m,n)}(\lambda) &= (\lambda - \alpha)^n(\lambda - \alpha)^m(\lambda - c_3)(\lambda - c_4) \\
    & + (\lambda - \alpha)^{m-1}(\lambda -c_5)(\lambda - c_6)(\lambda - \alpha)^{n-1}(\lambda - c_1)(\lambda - c_2) \nonumber \\ 
    & - \lambda (\lambda - \alpha)^n (\lambda - \alpha)^{m-1}(\lambda - c_5)(\lambda - c_6)\\
    & \\
    &=(\lambda - \alpha)^{m+n-2} \Biggl[(\lambda - \alpha)^2(\lambda - c_3)(\lambda - c_4)\\
    & + (\lambda - c_5)(\lambda - c_6)(\lambda - c_1)(\lambda - c_2) - \lambda(\lambda - \alpha)(\lambda - c_5)(\lambda - c_6)\Biggr]\\
    &=(\lambda - \alpha)^{m+n-2}f(\lambda)
\end{align*}
}%
where $f(\lambda) = x^4 - \alpha(m + n + 4)x^3 + [\alpha^2(mn + 2m + 2n + 5) + (2\alpha - 1)(m+n+1)]x^2 + [\alpha^2(m+n+2) + 2(2\alpha - 1)(m+1)(n+1)]x + (2\alpha - 1)[\alpha^2(m+1)(n+1) + mn(\alpha -1)^2]$, 
\begin{align*}
    &c_1 = \dfrac{1}{2}\left(\alpha(n+1) + \sqrt{\alpha^2(n+1)^2 + 4n(1-2\alpha)} \right), \\
    &c_2 = \dfrac{1}{2}\left(\alpha(n+1) - \sqrt{\alpha^2(n+1)^2 + 4n(1-2\alpha)} \right), \\
    &c_3 = \dfrac{1}{2}\left(\alpha(m+2) + \sqrt{\alpha^2(m+2)^2 + 4(m+1)(1-2\alpha)} \right), \\
    &c_4 = \dfrac{1}{2}\left(\alpha(m+2) - \sqrt{\alpha^2(m+2)^2 + 4(m+1)(1-2\alpha)} \right), \\
    &c_5 = \dfrac{1}{2}\left(\alpha(m+2) + \sqrt{\alpha^2((m+2)^2 - 4) + 2m(1-2\alpha)} \right) \text{ and }\\
    &c_6 = \dfrac{1}{2}\left(\alpha(m+2) - \sqrt{\alpha^2((m+2)^2 - 4) + 2m(1-2\alpha)} \right).
\end{align*}
\end{example}

\begin{rmk}
    In the example \ref{ex::double_star} if $m=n$, we get the balanced double star ($D_{n,n}$) and in that case the characteristic polynomial is
    \begin{equation*}
        P_{A_\alpha(D(n,n))}(\lambda) = (\lambda - \alpha)^{2n-2}g(\lambda),
    \end{equation*}
    where $g(\lambda) = (-\lambda^2 + \lambda(\alpha n + 3 \alpha -1) - 2\alpha^2 + \alpha + n)(-\lambda^2 + \lambda(\alpha n + \alpha + 1) -2n\alpha =\alpha +n)$, whose roots are
    \begin{align*}
        r_1 &= \dfrac{1}{2} \left( \alpha(n + 1) -1 -  \sqrt{\alpha^2(n+3)^2 - 2\alpha(4\alpha +n) - (2\alpha -1)(4n+1)}\right)\\
        r_2 &= \dfrac{1}{2} \left( \alpha(n + 1) -1 + \sqrt{\alpha^2(n+3)^2 - 2\alpha(4\alpha +n) - (2\alpha -1)(4n+1)}\right)\\
        r_3 &= \dfrac{1}{2} \left( \alpha(n + 1) + 1 - \sqrt{\alpha^2(n+1)^2 - 2n(3\alpha + 2) - (2\alpha -1)}\right)\\
        r_4 &= \dfrac{1}{2} \left( \alpha(n + 1) + 1 + \sqrt{\alpha^2(n+1)^2 - 2n(3\alpha + 2) - (2\alpha -1)}\right)
    \end{align*}
    So we can conclude that
    {\footnotesize
    \begin{align*}
       \sigma(A_\alpha(D_{n,n})) &= \Biggl\{ \dfrac{1}{2} \left( \alpha(n + 1) -1 -  \sqrt{\alpha^2(n+3)^2 - 2\alpha(4\alpha +n) - (2\alpha -1)(4n+1)}\right), \\
       & \dfrac{1}{2} \left( \alpha(n + 1) -1 + \sqrt{\alpha^2(n+3)^2 - 2\alpha(4\alpha +n) - (2\alpha -1)(4n+1)}\right), \\
       & \dfrac{1}{2} \left( \alpha(n + 1) + 1 - \sqrt{\alpha^2(n+1)^2 - 2n(3\alpha + 2) - (2\alpha -1)}\right), \\
       & \dfrac{1}{2} \left( \alpha(n + 1) + 1 + \sqrt{\alpha^2(n+1)^2 - 2n(3\alpha + 2) - (2\alpha -1)}\right), \alpha^{(2n-2)}
       \Biggr\} 
    \end{align*}
    }%
\end{rmk}

Let $G \cong K_{1,n+1}$ and $H \cong K_{1,m+1}$ with $m \neq n$. Consider $G \cdot H$ be a graph obtained from the coalescence identifying two vertices of degree $1$. Fig. \ref{fig::double_broom} shows $G \cdot H \cong B(3,n,m)$.
    \begin{figure}[!h]
        \centering
        \includegraphics[width=0.4\textwidth]{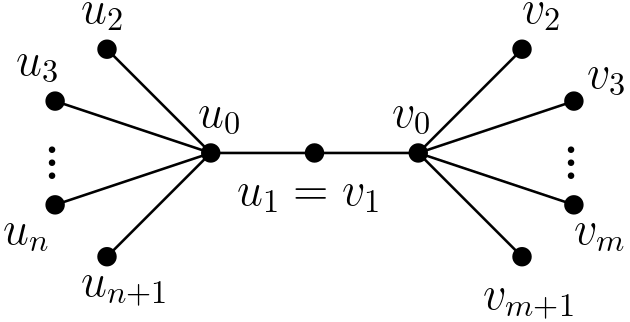}
         \vspace{0.5cm}
        \centering
        \caption{$B(3,n,m)$.}
        \label{fig::double_broom}
    \end{figure}

\begin{example}
Consider $B(3,n,m)$. Applying Corollary \ref{cor::submatrix_star}, Proposition \ref{prop::star_spectrum} and Theorem \ref{theo::charpol_coalescence} we have

\begin{align*}
     P_{A_\alpha(G \cdot H)}(\lambda) &= (\lambda - \alpha)^{m+n-2} ( (\lambda - \alpha)(\lambda - c_1)(\lambda - c_2)(\lambda - c_7)(\lambda - c_8)\\
     & \qquad + (\lambda - \alpha)(\lambda - c_5)(\lambda - c_6)(\lambda - c_3)(\lambda - c_4)\\
     & \qquad - \lambda(\lambda - c_5)(\lambda - c_6)(\lambda - c_7)(\lambda - c_8)) \\
     & = (\lambda - \alpha)^{m+n-2} f(\lambda)
 \end{align*}
 where $f(\lambda) = \lambda^5 -\alpha(m+n+6)\lambda^4 + ((m+4)(n+4)-4\alpha^2 + (2\alpha-1)(m+n+2))\lambda^3 - \left(  \dfrac{a(4a^2(m + 2)(n + 2) + 4a^2 + (2a - 1)(4mn + 9m + 9n + 12))}{2} \right)\lambda^2 +\Biggl(\alpha^4(m+n+3) + 3\alpha^3(2mn +3m +3n +4) -\dfrac{\alpha}{2}(5m+5n+12)-\dfrac{(4\alpha - 1)(3mn + 2m + 2n)}{2} \Biggr)\lambda -\alpha(2\alpha - 1)^2(2mn+m+n) + 2\alpha^2(m+n+2)$, 
 \begin{align*}
    &c_1=\dfrac{1}{2}\left(\alpha(n+2) + \sqrt{\alpha^2(n+2)^2 + 4(n+1)(1-2\alpha)} \right), \\ 
    &c_2=\dfrac{1}{2}\left(\alpha(n+2) - \sqrt{\alpha^2(n+2)^2 + 4(n+1)(1-2\alpha)} \right), \\
    &c_3 =\dfrac{1}{2}\left(\alpha(m+2) + \sqrt{\alpha^2(m+2)^2 + 4(m+1)(1-2\alpha)} \right),\\  
    &c_4 =\dfrac{1}{2}\left(\alpha(m+2) - \sqrt{\alpha^2(m+2)^2 + 4(m+1)(1-2\alpha)} \right),\\
    &c_5 = \dfrac{1}{2}\left(\alpha(n+2) + \sqrt{\alpha^2((n+2)^2 - 4) + 2n(1-2\alpha)} \right),\\
    &c_6 = \dfrac{1}{2}\left(\alpha(n+2) - \sqrt{\alpha^2((n+2)^2 - 4) + 2n(1-2\alpha)} \right),\\ 
    &c_7 = \dfrac{1}{2}\left(\alpha(m+2) + \sqrt{\alpha^2((m+2)^2 - 4) + 2m(1-2\alpha)} \right) \text{ and } \\
    &c_8 = \dfrac{1}{2}\left(\alpha(m+2) - \sqrt{\alpha^2((m+2)^2 - 4) + 2m(1-2\alpha)} \right).
\end{align*}
In particular case, when $m = n$ we have
 $$P_{A_\alpha(G \cdot G)}(\lambda) = (\lambda - \alpha)^{2n-2} g(\lambda),$$
where $g(\lambda) = \lambda^5 -2\alpha(n + 3)\lambda^4 + (\alpha^2(n+2)(n+6) + 2(2\alpha - 1)(n + 1))\lambda^3 + (-2\alpha^3(n^2 + 4n + 5) -\alpha(2\alpha - 1)(2n^2 + 9n + 6))\lambda^2 + \Biggl( \dfrac{(4\alpha^2n + 6\alpha^22 + 2\alpha n - n)(2\alpha^2 + 6\alpha n + 8\alpha - 3n - 4)}{4} \Biggr) \lambda$
$ -\alpha(2\alpha - 1)(n + 1)(2\alpha^2 + 2\alpha n - n).$
\end{example}

Next result presents the $A_\alpha(l(G))$-characteristic polynomial when $G$ is a regular graph in function of $A_\alpha(G)$ and $A(G)$-characteristic polynomial.

\begin{theorem} \label{theo::linegraph}
	Let $G$ be a $r$-regular graph with $n$ vertices and $m$ edges such that $r \geq 2$ and $\alpha \in [0,1)$. Then
	\begin{equation*} \label{eq1::linegraph}
	    P_{A_\alpha(l(G))}(\lambda) = (\lambda - 2r \alpha + 2)^{m-n}P_{A_\alpha(G)}(\lambda - r + 2)
	\end{equation*}
	and
	\begin{equation*} \label{eq2::linegraph}
	    P_{A_\alpha(l(G))}(\lambda) = (\lambda - 2r \alpha + 2)^{m-n}(1-\alpha)^n P_{A(G)}\left( \dfrac{\lambda - r(\alpha+1) + 2}{1-\alpha} \right)
	\end{equation*}
\end{theorem}

\begin{proof}
	Let $G$ be a $r$-regular graph with $n$ vertices and $m$ edges and $B$ its incident matrix. From Lemma \ref{lemma::incident_BBT}  we have 
	\begin{align} \label{eq::BBT_Aalpha}
		(1-\alpha)BB^T &= (1-\alpha)D(G) + (1-\alpha)A(G) \nonumber \\
		&= (1-\alpha)D(G) + A_\alpha(G) - \alpha D(G) \nonumber\\
		&= (1 - 2 \alpha)D(G) + A_\alpha(G).
	\end{align} 
	On the other hand $A_\alpha(l(G)) = \alpha D(l(G)) + (1- \alpha)A(l(G))$. Consider the matrices $U$ e $V$ as follows
	$$ U = 
	\begin{bmatrix}
		\lambda I_n & -B\\
		0 & I_m 
	\end{bmatrix} \mathrm{ \ \ e \ \ }
	V = 
	\begin{bmatrix}
		I_n & B\\
		(1-\alpha)B^T & \lambda I_m 
	\end{bmatrix}.
	$$
	Then,
	$$UV = 
	\begin{bmatrix}
		\lambda I_n & -B\\
		0 & I_m 
	\end{bmatrix}
	\begin{bmatrix}
		I_n & B\\
		(1-\alpha)B^T & \lambda I_m 
	\end{bmatrix} = 
	\begin{bmatrix}
		\lambda I_n - (1-\alpha)BB^T& 0\\
		(1-\alpha)B^T & \lambda I_m 
	\end{bmatrix}$$
	and
	$$VU = 
	\begin{bmatrix}
		I_n & B\\
		(1-\alpha)B^T & \lambda I_m 
	\end{bmatrix}
	\begin{bmatrix}
		\lambda I_n & -B\\
		0 & I_m 
	\end{bmatrix} = 
	\begin{bmatrix}
		\lambda I_n & 0\\
		(1-\alpha) \lambda B^T &  \lambda I_m - (1-\alpha)B^TB
	\end{bmatrix}
	$$
    As $\vert UV \vert = \vert VU \vert$, we have
	\begin{equation} \label{eq::linegraph_det_identity}
		\lambda^m \vert \lambda I_n -(1-\alpha)BB^T \vert = \lambda^n \vert \lambda I_m - (1-\alpha)B^TB \vert.
	\end{equation} 
By definition, 
 	\begin{equation*}
		P_{A_\alpha(l(G))}(\lambda) = \vert \lambda I_m - A_\alpha(l(G))\vert = \vert \lambda I_m -\alpha D(l(G)) - (1-\alpha)A(l(G))\vert
	\end{equation*} and using Lemma \ref{lemma::incident_BTB}, we have
\begin{align*}
	P_{A_\alpha(l(G))}(\lambda) &= \vert\lambda I_m - \alpha D(l(G)) + (1 - \alpha)(2I_m - B^TB) \vert\\
	&= \vert(\lambda + 2(1-\alpha))I_m - (1-\alpha)B^TB - \alpha \underbrace{D(l(G))}_{2(r-1)I_m} \vert\\
	&=\vert(\lambda + 2 -2\alpha -2\alpha r + 2\alpha)I_m - (1-\alpha)B^TB\vert\\
	&= \vert(\lambda -2r \alpha + 2)I_m - (1-\alpha)B^TB\vert
\end{align*}
and applying equation (\ref{eq::linegraph_det_identity}),
\begin{equation} \label{eq1::linegraph_det_identity}
		P_{A_\alpha(l(G))}(\lambda) = (\lambda -2r \alpha + 2)^{m-n} \vert(\lambda -2r \alpha + 2)I_n - (1-\alpha)BB^T\vert 
\end{equation}
Finally, from equation (\ref{eq::BBT_Aalpha}), we have 
\begin{align*}
	P_{A_\alpha(l(G))}(\lambda) &= (\lambda -2r \alpha + 2)^{m-n} \vert(\lambda -2r \alpha + 2)I_n - [(1 - 2 \alpha)D + A_\alpha(G)]\vert\\
	&=(\lambda -2r \alpha + 2)^{m-n} \vert(\lambda -2r \alpha + 2)I_n - (1-2\alpha)rI_n - A_\alpha(G)\vert\\
	&=(\lambda -2r \alpha + 2)^{m-n} \vert(\lambda -r + 2)I_n - A_\alpha(G)\vert\\
	&=(\lambda -2r \alpha + 2)^{m-n}P_{A_\alpha(G)}(\lambda -r + 2)
\end{align*}
From (\ref{eq1::linegraph_det_identity}), we know that
\begin{align*}
     P_{A_\alpha(l(G))}(\lambda) &= (\lambda -2r \alpha + 2)^{m-n} \vert(\lambda -2r \alpha + 2)I_n - (1-\alpha)BB^T\vert\\
     &=(\lambda -2r \alpha + 2)^{m-n} \vert(\lambda -2r \alpha + 2)I_n - (1-\alpha)(A(G) + rI_n)\vert\\
     &=(\lambda -2r \alpha + 2)^{m-n} \vert(\lambda -r \alpha -r + 2)I_n - (1-\alpha)A(G)\vert\\
     &=(\lambda - 2r \alpha + 2)^{m-n}(1-\alpha)^n P_{A(G)}\left( \dfrac{\lambda - r(\alpha+1) + 2}{1-\alpha} \right)
\end{align*}
and the result follows.
\end{proof}

\begin{example} \label{example::complete_graph}
Consider $G \cong K_n$. From Proposition \ref{prop::complete_bipartite_spectrum}, $\sigma(A_\alpha(K_n)) = \{n-1, \alpha n - 1^{(n-1)} \}$. So, from Theorem \ref{theo::linegraph}, $\sigma(A_\alpha(l(K_n))) = \Biggl\{ 2n-4, n(\alpha + 1) - 4^{(n-1)}$, $2\alpha(n-1) - 2^{(\frac{n(n-3)}{2})} \Biggr\}$.
\end{example}

\begin{example}
It is easy to see that $l(K_5)$ is a graph of order $10$ and $6$-regular. From Example \ref{example::complete_graph} we have that $\sigma(A_\alpha(l(K_5))) = \{6, 5 \alpha + 1^{(4)}, 8 \alpha - 2^{(5)} \}$. Moreover we know that $\overline{l(K_5)}$ is the Petersen graph and from Proposition \ref{prop::charpol_graph_complement} we have $\sigma(A_\alpha(\overline{l(K_5)})) = \{ 3, 5 \alpha - 2^{(4)}, 2 \alpha + 1^{(5)}\}$. Taking $n=4$ in the example \ref{example::complete_graph} we obtain $l(K_4)$ which is the octahedral graph shows in Fig. \ref{fig::octahedral_graph} and consequently $\sigma(A_\alpha(l(K_4))) = \{ 4, 4 \alpha^{(3)}, 6\alpha - 2^{(2)}\}$. 
\end{example}

\begin{figure}[!h]%
    \centering
        \includegraphics[width=0.4\textwidth]{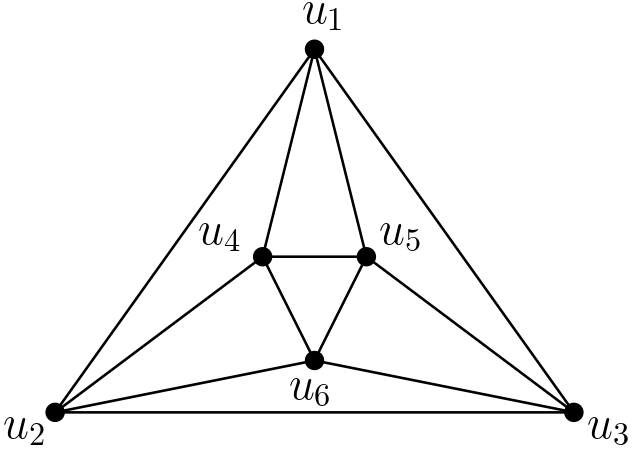}
         \vspace{0.5cm}
        \centering
        \caption{Octahedral graph $(l(K_4))$.}
        \label{fig::octahedral_graph}
\end{figure}

Theorem \ref{theo::charpol_semirregular_linegraph} and Corollary \ref{cor::charpol_semirregular_linegraph} present $A_\alpha(G)$-characteristic polynomial when $G$ is a semi-regular bipartite graph.

\begin{theorem} \label{theo::charpol_semirregular_linegraph}
	Let $G$ be a semi-regular bipartite graph with positive integers parameters $n_1$, $n_2$, $r_1$ and $r_2$, such that $ n_1 \geq n_2$ and  $\alpha \in [0,1)$. Then,
	\begin{equation*}
		P_{A_\alpha(l(G))}(\lambda) = (1-\alpha)^{\gamma}(\lambda - \alpha (r_1 + r_2) + 2)^\beta P_{A(G)}(\sqrt{\alpha_1\alpha_2})\sqrt{\left(\dfrac{\alpha_1}{\alpha_2}\right)^{n_1-n_2}},
	\end{equation*}
	where $\alpha_1= \displaystyle \frac{\lambda - \alpha r_2 - r_1 + 2}{1-\alpha}$, $\alpha_2= \displaystyle \frac{\lambda - \alpha r_1 - r_2 + 2}{1-\alpha}$, $\beta = n_1r_1 - n$ and $\gamma = \dfrac{n_2r_2 - n_1r_1 + 2n}{2}$.
\end{theorem}
\begin{proof}
	Let $G$ be a semi-regular bipartite graph with partitions $V_1$ and $V_2$ such that $\vert V_1 \vert = n_1$ e $\vert V_2 \vert = n_2$, where $n_1 \geq n_2$ and $n_1 + n_2  =  n$. We know that for all $v \in V_1$, $d(v) = r_1$ and for all $u \in V_2$, $d(u) = r_2$. Moreover, we  know that $\displaystyle A(l(G)) =  \frac{1}{1-\alpha} \left( A_\alpha(l(G)) -  \alpha D(l(G)) \right)$.
	
	As $m_G = n_{l(G)} = \dfrac{n_1r_1 + n_2r_2}{2}$, follows that
	\begin{align} 
		\displaystyle \lambda I_{n_{l(G)}} - A(l(G)) &= \lambda I_{n_{l(G)}} - \frac{1}{1-\alpha} A_\alpha(l(G)) + \frac{\alpha}{1-\alpha} D(l(G))  \nonumber \\
		\lambda I_{n_{l(G)}} - A(l(G)) &= \frac{(1-\alpha)\lambda I_{n_{l(G)}}}{1-\alpha} - \frac{1}{1-\alpha} A_\alpha(l(G)) + \frac{\alpha}{1-\alpha}D(l(G)). \label{eq::algebraic_manipulation}
	\end{align}
	
	Computing the determinant in (\ref{eq::algebraic_manipulation}) we get
	\begin{align*}
		\displaystyle \left \vert \lambda I_{n_{l(G)}} - A(l(G)) \right \vert = \left\vert \frac{(1-\alpha)\lambda I_{n_{l(G)}}}{1-\alpha} - \frac{1}{1-\alpha} A_\alpha(l(G)) + \frac{\alpha}{1-\alpha}D(l(G)) \right \vert \Leftrightarrow
	\end{align*}
	\begin{align*}
		(1 - \alpha)^{n_{l(G)}}\left\vert \lambda I_{n_{l(G)}} - A(l(G)) \right\vert = \left \vert (1-\alpha)\lambda I_{n_{l(G)}} -  A_\alpha(l(G)) +  \alpha D(l(G)) \right \vert. 
	\end{align*}
	
	Since $G$ is semi-regular bipartite we have $D(l(G)) = (r_1 + r_2 -2)I_{n_{l(G)}}$, and this implies that

\begin{align*}
	\displaystyle \left\vert ((1-\alpha)\lambda + \alpha(r_1 + r_2 -2)) I_{n_{l(G)}} -  A_\alpha(l(G)) \right \vert =& (1-\alpha)^{n_{l(G)}}\left \vert \lambda I_{n_{l(G)}} - A(l(G)) \right \vert 
\end{align*}
Then,
\begin{align} 
    P_{A_\alpha(l(G))}((1-\alpha)\lambda + \alpha(r_1 + r_2 -2)) &= (1-\alpha)^{n_{l(G)}} P_{A(l(G))}(\lambda) \Leftrightarrow \nonumber \\
	P_{A_\alpha(l(G))}(\lambda) &= (1-\alpha)^{n_{l(G)}} P_{A(l(G))}\left(\frac{\lambda - \alpha(r_1 + r_2 -2)}{1-\alpha} \right) \label{eq1::charpol_semiregular}
\end{align}
	Applying Theorem \ref{theo::charpol_semi_regular} in (\ref{eq1::charpol_semiregular}), we obtain 
	\begin{align*}
		P_{A_\alpha(l(G))}(\lambda) = (1-\alpha)^{\gamma}(\lambda - \alpha(r_1 + r_2) + 2)^\beta P_{A(G)}(\sqrt{\alpha_1\alpha_2})\sqrt{\left(\dfrac{\alpha_1}{\alpha_2}\right)^{n_1-n_2}},
	\end{align*}
	where $\alpha_1= \displaystyle \frac{\lambda - \alpha r_2 - r_1 + 2}{1-\alpha}$, $\alpha_2= \displaystyle \frac{\lambda - \alpha r_1 - r_2 + 2}{1-\alpha}$, $\beta = n_1r_1 - n_1 - n_2=n_1r_1-n$ and $\gamma = \dfrac{n_2r_2 - n_1r_1 + 2(n_1 + n_2)}{2} = \dfrac{n_2r_2 - n_1r_1 + 2n}{2}$ and the result follows.
\end{proof}

\begin{corollary} \label{cor::charpol_semirregular_linegraph}
	Let $G$ be a semi-regular bipartite graph with  positive integer parameters  $n_1,n_2,r_1 \text{ and }r_2$, such that $n_1 \geq n_2$ and $\alpha \in [0,1)$. If $\lambda_1(A(G)), \ldots, \lambda_{n_2}(A(G))$ are the $n_2$ largest eigenvalues of $A(G)$, then
	
	{\footnotesize
	\begin{equation}
	   P_{A_\alpha(l(G))}(\lambda)= (\lambda - \alpha(r_1 + r_2) + 2)^{\beta+1}(\lambda - (r_1+r_2) + 2)(1-\alpha)^{\gamma - 2}\alpha_1^{n_1 - n_2} \prod_{i=2}^{n_2}(\alpha_1\alpha_2 - \lambda_i^2)
	\end{equation}
	}%
	where $\alpha_1= \displaystyle \frac{\lambda - \alpha r_2 - r_1 + 2}{1-\alpha}$, $\alpha_2= \displaystyle \frac{\lambda - \alpha r_1 - r_2 + 2}{1-\alpha}$, $\beta = n_1r_1 - n$ and $\gamma = \dfrac{n_2r_2 - n_1r_1 + 2n}{2}$.
\end{corollary}
\begin{proof}
	Let $G$ be a semi-regular bipartite graph with  positive integer parameters $n_1$, $n_2$, $r_1$ and $r_2$ such that $n_1 \geq n_2$. Suppose that  the $n_2$ largest eigenvalues of $A(G)$ are $\lambda_1(A(G)), \ldots, \lambda_{n_2}(A(G))$. We know that  
	$A(G) =  \begin{bmatrix} 
		\bigzero & B\\
		B^T & \bigzero
	\end{bmatrix}$
	and from Theorem \ref{prop::semi_regular_spectrum} follows that $\lambda_1(A(G)) = \sqrt{r_1r_2}$. Moreover, as the rank of $B$ is at most  $n_2$,  $A(G)$ has  $0$ as eigenvalue with multiplicity, at least $n_1-n_2$. So, from previous Theorem, \ref{theo::charpol_semirregular_linegraph}, we have
	\begin{equation} \label{eq1::cor_semi_regular}
		P_{A_\alpha(l(G))}(\lambda) = (1-\alpha)^{\gamma}(\lambda - \alpha(r_1 + r_2) + 2)^\beta \alpha_1^{n_1 - n_2}\prod_{i=1}^{n_2}(\alpha_1\alpha_2 - \lambda_i^2)
	\end{equation}
	where $\alpha_1= \displaystyle \frac{\lambda - \alpha r_2 - r_1 + 2}{1-\alpha}$, $\alpha_2= \displaystyle \frac{\lambda - \alpha r_1 - r_2 + 2}{1-\alpha}$, $\beta = n_1r_1 - n$ and $\gamma = \dfrac{n_2r_2 - n_1r_1 + 2n}{2}$. As $\alpha_1\alpha_2 - \lambda_1^2 = \alpha_1\alpha_2 - r_1r_2$, we have
	\begin{equation} \label{eq2::cor_semi_regular}
		\displaystyle \alpha_1\alpha_2 - r_1r_2 = \frac{(\lambda - (r_1 + r_2) +2)( \lambda -\alpha(r_1+r_2) +2) }{(1-\alpha)^2}
	\end{equation}
	
	Substituting (\ref{eq2::cor_semi_regular}) in (\ref{eq1::cor_semi_regular}), the result follows. 
\end{proof}
 The next theorems present the $A_\alpha$-characteristic polynomials of  $S(G)$, $R(G)$, $Q(G)$ and $T(G)$, when $G$ is a $r$-regular graph. 

\begin{theorem}\label{theo1::subdivision_graph}
Let $G$ be a graph with $n$ vertices and $m$ edges. If $G$ is $r$-regular and $\alpha \in [0,1)$, then
\begin{equation*}\label{eq2::subdivision_graph}
    P_{A_\alpha(S(G))}(\lambda) = (\lambda - 2 \alpha)^{m-n}(1-\alpha)^{n}P_{A_\alpha(G)}\left(\dfrac{\lambda^2 -\alpha(r+2)\lambda + r(3\alpha -1)}{(1-\alpha)} \right)
\end{equation*}
and
\begin{equation*}\label{eq1::subdivision_graph}
    P_{A_\alpha(S(G))}(\lambda) = (\lambda - 2 \alpha)^{m-n}(1-\alpha)^{2n}P_{A(G)}\left(\dfrac{\lambda^2 -\alpha(r+2)\lambda + r(\alpha^2 + 2\alpha -1)}{(1-\alpha)^2} \right)
\end{equation*}
\end{theorem}
\begin{proof}
Let $G$ be a $r$-regular graph with $n$ vertices, $m$ edges and $B$ its incident matrix. We know that $A_\alpha(S(G)) = \alpha D(S(G)) + (1-\alpha)A(S(G)) $. As $ A(S(G)) = 
\begin{bmatrix}
    \bigzero & B^T\\
    B & \bigzero
\end{bmatrix}$, we have
$$P_{A_\alpha(S(G))}(\lambda) = \vert \lambda I_{m+n} - A_\alpha(S(G)) \vert = \begin{vmatrix}
    (\lambda - 2\alpha)I_m & -(1-\alpha)B^T\\
    -(1-\alpha)B & (\lambda - \alpha r)I_n 
    \end{vmatrix} 
$$
From Theorem \ref{theo::inv_block_matrix}, follows that
\begin{align}
    P_{A_\alpha(S(G))}(\lambda) &=(\lambda - 2 \alpha)^m \left \vert (\lambda - \alpha r)I_n - \dfrac{(1-\alpha)^2}{\lambda - 2 \alpha}BB^T \right \rvert \nonumber \\
    &= (\lambda - 2 \alpha)^{m-n} \left \vert (\lambda - \alpha r)(\lambda - 2 \alpha)I_n - (1-\alpha)^2BB^T \right \rvert. \label{eq3::subdivision_graph}
\end{align}
From Lemma \ref{lemma::incident_BBT} and (\ref{eq3::subdivision_graph}), we have 
\begin{align*}
    P_{A_\alpha(S(G))}(\lambda) &= (\lambda - 2 \alpha)^{m-n} \vert ((\lambda - \alpha r)(\lambda - 2\alpha)- r(1-\alpha)^2)I_n - (1-\alpha)^2A(G) \vert\\
    &=(\lambda - 2 \alpha)^{m-n}(1-\alpha)^{2n}P_{A(G)}\left(\dfrac{\lambda^2 -\alpha(r+2)\lambda + r(\alpha^2 + 2\alpha -1)}{(1-\alpha)^2} \right)
\end{align*}
 Moreover 
\begin{align}
    (1-\alpha)^2BB^T &= (1-\alpha)^2r I_n + (1-\alpha)^2 A(G) \nonumber \\
    &=(1-\alpha)^2r I_n + (1-\alpha)( A_\alpha(G) - \alpha r I_n) \nonumber \\
    &=(1-3\alpha+2\alpha^2)r I_n + (1-\alpha)A_\alpha(G). \label{eq::S(G)}
\end{align}
From (\ref{eq3::subdivision_graph}) and (\ref{eq::S(G)}) we have
\begin{align*}
    P_{A_\alpha(S(G))}(\lambda) &= (\lambda - 2 \alpha)^{m-n} \vert (\lambda^2 - (2+r)\alpha \lambda + r(3\alpha - 1))I_n -(1-\alpha)A_\alpha(G) \vert\\
    &= (\lambda - 2 \alpha)^{m-n}(1-\alpha)^{n}P_{A_\alpha(G)}\left(\dfrac{\lambda^2 -\alpha(r+2)\lambda + r(3\alpha -1)}{(1-\alpha)} \right)
\end{align*}
and the result follows.
\end{proof}

\begin{theorem} \label{theo1::R(G)}
Let $G$ be a graph with $n$ vertices and $m$ edges. If $G$ is $r$-regular and $\alpha \in [0,1)$, then
{\footnotesize
\begin{equation*}\label{eq2::R(G)}
     P_{A_\alpha(R(G))}(\lambda) = (\lambda - 2\alpha)^{m-n}(\lambda - 3\alpha +1)^n P_{A_\alpha(G)}\left(\dfrac{\lambda^2 -\alpha(r+2)\lambda + r(-3\alpha +1)}{\lambda - 3\alpha +1}\right)
\end{equation*}
}
and
{\footnotesize
\begin{equation*}\label{eq1::R(G)}
    P_{A_\alpha(R(G))}(\lambda) = (\lambda - 2\alpha)^{m-n}((1-\alpha)(\lambda - 3\alpha +1))^n P_{A(G)}\left(\dfrac{\lambda^2 -2\alpha(r+1)\lambda + r(3\alpha^2 + 2\alpha -1)}{(1-\alpha)(\lambda - 3\alpha +1)}\right)
\end{equation*}
}
\end{theorem}
\begin{proof}
Let $G$ be a $r$-regular graph with $n$ vertices, $m$ edges and $B$ its incident matrix. It is known that $A_\alpha(R(G)) = \alpha D(R(G)) + (1-\alpha)A(R(G))$. As  $A(R(G)) = \begin{bmatrix}
\bigzero & B^T \\
B & A(G)
\end{bmatrix}$ we have
$$P_{A_\alpha(R(G))}(\lambda) = \vert \lambda I_{n+m} - A_\alpha(R(G)) \vert =\begin{vmatrix}
    (\lambda - 2\alpha)I_m & -(1-\alpha)B^T \\
    -(1-\alpha)B & (\lambda -2\alpha r)I_n - (1-\alpha)A(G)
    \end{vmatrix} $$

From Theorem \ref{theo::inv_block_matrix} follows that
\begin{align}
    P_{A_\alpha(R(G))}(\lambda) &=(\lambda - 2\alpha)^m \left \vert(\lambda - 2\alpha r)I_n - (1-\alpha)A(G) - \dfrac{(1-\alpha)^2BB^T}{\lambda - 2\alpha} \right \vert \nonumber \\
    &=(\lambda - 2\alpha)^{m-n} \left \vert(\lambda - 2\alpha)((\lambda - 2\alpha r)I_n - (1-\alpha)A(G)) - (1-\alpha)^2BB^T \right \vert \label{eq3::R(G)}
\end{align}
From (\ref{eq3::R(G)}) and Lemma \ref{lemma::incident_BBT}, we get
{\footnotesize
\begin{multline*}
    P_{A_\alpha(R(G))}(\lambda) = (\lambda - 2\alpha)^{m-n} \left \vert(\lambda - 2\alpha)((\lambda - 2\alpha r)I_n - (1-\alpha)A(G)) - (1-\alpha)^2(A(G) + r I_n) \right \vert \\
    = (\lambda - 2\alpha)^{m-n} \left \vert(\lambda^2 -2\alpha(r+1)\lambda + r(3\alpha^2 + 2\alpha -1) )I_n - ( (1-\alpha)(\lambda - 3\alpha +1))A(G) \right \vert \\
    = (\lambda - 2\alpha)^{m-n}((1-\alpha)(\lambda - 3\alpha +1))^n P_{A(G)}\left(\dfrac{\lambda^2 -2\alpha(r+1)\lambda + r(3\alpha^2 + 2\alpha -1)}{(1-\alpha)(\lambda - 3\alpha +1)}\right)
\end{multline*}
}%
It is possible to rewrite $P_{A_\alpha(R(G))}(\lambda)$ using $A_\alpha(G)$ the following way
$$P_{A_\alpha(R(G))}(\lambda) = \vert \lambda I_{n+m} - A_\alpha(R(G)) \vert =\begin{vmatrix}
    (\lambda - 2\alpha)I_m & -(1-\alpha)B^T \\
    -(1-\alpha)B & (\lambda -\alpha r)I_n - A_\alpha(G)
    \end{vmatrix}
$$
So, from Theorem \ref{theo::inv_block_matrix} 
\begin{align},
   P_{A_\alpha(R(G))}(\lambda) &=(\lambda - 2\alpha)^m \left \vert(\lambda - \alpha r)I_n - A_\alpha(G) - \dfrac{(1-\alpha)^2BB^T}{\lambda - 2\alpha} \right \vert \nonumber \\
    &=(\lambda - 2\alpha)^{m-n} \left \vert(\lambda - 2\alpha)((\lambda - \alpha r)I_n - A_\alpha(G)) - (1-\alpha)^2BB^T \right \vert \label{eq4::R(G)}
\end{align}
From Lemma \ref{lemma::incident_BBT} and (\ref{eq4::R(G)}), we have
{\footnotesize
\begin{align*}
    P_{A_\alpha(R(G))}(\lambda) &= (\lambda - 2\alpha)^{m-n} \left \vert(\lambda^2 -(r+2)\alpha \lambda -3\alpha r + r)I_n - (\lambda -3\alpha +1)A_\alpha(G) \right \vert \\
    &= (\lambda - 2\alpha)^{m-n}(\lambda - 3\alpha +1)^n P_{A_\alpha(G)}\left(\dfrac{\lambda^2 -\alpha(r+2)\lambda + r(-3\alpha +1)}{\lambda - 3\alpha +1}\right)
\end{align*}
}%
and the result follows.
\end{proof}

\begin{theorem}\label{theo::Q(G)}
Let $G$ be a graph with $n$ vertices and $m$ edges. If $G$ is $r$-regular and $\alpha \in [0,1]$, then
{\footnotesize
\begin{equation*} \label{eq1::Q(G)}
    P_{A_\alpha(Q(G))}(\lambda) = (\lambda - \alpha r)^{n-m}(\lambda - (r+1)\alpha +1)^m P_{A_\alpha(l(G))} \left(\dfrac{\lambda^2 -\alpha r \lambda -\alpha^2 + 2\alpha(r+1) -2}{\lambda -(r+1)\alpha +1} \right)
\end{equation*}
}%
\end{theorem}
\begin{proof}
Let $G$ be a $r$-regular graph with $n$ vertices, $m$ edges and $B$ its incident matrix. From definition,  $A_\alpha(Q(G)) = \alpha D(Q(G)) + (1-\alpha)A(Q(G))$, where $A(Q(G)) = \begin{bmatrix}
\bigzero & B \\
B^T & A(l(G))
\end{bmatrix}$, so 
$$P_{A_\alpha(Q(G))}(\lambda) = \vert \lambda I_{n+m} - A_\alpha(Q(G)) \vert =\begin{vmatrix}
    (\lambda - r\alpha)I_n & -(1-\alpha)B \\
    -(1-\alpha)B^T & \lambda I_m - A_\alpha(l(G))
    \end{vmatrix}
$$
From Theorem \ref{theo::inv_block_matrix} we have
\begin{align}
    P_{A_\alpha(Q(G))}(\lambda) &= (\lambda - \alpha r)^n \left \vert \lambda I_m - A_\alpha(l(G)) - \dfrac{(1-\alpha)^2 B^TB}{\lambda - \alpha r} \right \vert \nonumber \\
    &= (\lambda - \alpha r)^{n-m}\vert (\lambda - \alpha r)(\lambda I_m - A_\alpha(l(G))) - (1-\alpha)^2B^TB \vert \label{eq2::Q(G)}
\end{align}
Applying Lemma \ref{lemma::incident_BTB},  we get
\begin{align}
    (1-\alpha)^2B^TB &= (1-\alpha)^2A(l(G)) + 2(1-\alpha)^2I_m \nonumber \\
    &= (1-\alpha)(A_\alpha(l(G)) - 2\alpha(r-1) I_m) + 2(1-\alpha)^2I_m \nonumber \\
    &= (1-\alpha)A_\alpha(l(G)) + (\alpha^2 -2\alpha(r+1) + 2)I_m \label{eq3::Q(G)}
\end{align}
From (\ref{eq2::Q(G)}) and (\ref{eq3::Q(G)}),
{\footnotesize 
\begin{align*}
     P_{A_\alpha(Q(G))}(\lambda) &= (\lambda - \alpha r)^{n-m} \left \vert (\lambda^2 -\alpha r \lambda - \alpha^2 + 2\alpha(r+1) - 2) I_m - (\lambda - \alpha(r+1) +1)A_\alpha(l(G)) \right \vert \\
     &= (\lambda - \alpha r)^{n-m}(\lambda - (r+1)\alpha +1)^m P_{A_\alpha(l(G))} \left(\dfrac{\lambda^2 -\alpha r \lambda -\alpha^2 + 2\alpha(r+1) -2}{\lambda -(r+1)\alpha +1} \right).
\end{align*}
}%
\end{proof}

As consequence the Theorems \ref{theo::Q(G)} and  \ref{theo::linegraph}, we obtain the Corollary \ref{cor::Q(G)}

\begin{corollary} \label{cor::Q(G)}
    Let $G$ be a graph with $n$ vertices and $m$ edges. If $G$ is $r$-regular and $\alpha \in [0,1)$, then
{\footnotesize
\begin{multline*}
    P_{A_\alpha(Q(G))}(\lambda) =\left(\dfrac{\lambda^2 + (2-3\alpha r)\lambda + \alpha^2(2r^2-1) - 2r\alpha(1-\alpha)}{(\lambda - \alpha r)(\lambda - \alpha(r+1) + 1)} \right)^{m-n}(\lambda - (r+1)\alpha +1)^m \\
    \cdot P_{A_\alpha(G)}\left(\dfrac{\lambda^2 - (\alpha r + r -2)\lambda -\alpha^2 +\alpha r^2 + \alpha r - r}{\lambda - \alpha(r+1) +1} \right)
\end{multline*}
}
and
{\footnotesize
\begin{multline*}
    P_{A_\alpha(Q(G))}(\lambda) = \left(\dfrac{\lambda^2 + (2-3\alpha r)\lambda + \alpha^2(2r^2-1) - 2r\alpha(1-\alpha)}{(\lambda - \alpha r)(\lambda - \alpha(r+1) + 1)} \right)^{m-n}(\lambda - (r+1)\alpha +1)^m \\
     \cdot (1-\alpha)^n P_{A(G)}\left(\dfrac{\lambda^2 - (2-r(2\alpha +1))\lambda + r(\alpha+1)(r\alpha + \alpha -1) -\alpha^2 }{(1-\alpha)(\lambda - \alpha(r+1) +1)} \right)
\end{multline*}
}%
\end{corollary}

The next results relate the eigenvalues of  $A_\alpha (T(G))$  with the eigenvalues of $A_\alpha(G)$ and $A(G)$.

\begin{theorem} \label{theo1::T(G)}
Let $G$ be a $r$-regular graph ($r > 1$) with $n$ vertices and $m$ edges. If the eigenvalues of $A_\alpha(G)$ are $\lambda_1(A_\alpha(G)), \ldots, \lambda_n(A_\alpha(G))$ and $\alpha \in [0,1]$, then $A_\alpha(T(G))$ has $m-n$ eigenvalues equal to $2\alpha(r+1) - 2$ and the others $2n$ eigenvalues are
{\footnotesize
\begin{equation*} \label{eq1::T(G)}
    \dfrac{1}{2}\Biggl(2(\alpha + \lambda_i(A_\alpha(G)) - 1) + r(\alpha + 1) \pm 
    \sqrt{(\alpha - 1)(\alpha(r+2)^2 -r^2 - 4(1+\lambda_i(A_\alpha(G))) )}\Biggr) 
\end{equation*}
}%
for $i = 1, \ldots , n.$
\end{theorem}
\begin{proof}
Let $G$ be a $r$-regular graph with $n$ vertices, $m$ edges and $B$ its incident matrix. As $A_\alpha(T(G)) = \alpha D(T(G)) + (1-\alpha)A(T(G))$ and  $A(T(G)) = \begin{bmatrix}
A(G) & B \\
B^T & A(l(G))
\end{bmatrix}$, we have
\begin{align*}
   P_{A_\alpha(T(G))}(\lambda) &= \vert \lambda I_{n+m} - A_\alpha(T(G)) \vert \\ 
   &=\begin{vmatrix}
    (\lambda - 2r\alpha)I_n -(1-\alpha)A(G) & -(1-\alpha)B \\
    -(1-\alpha)B^T & (\lambda-2r\alpha) I_m -(1-\alpha)A_\alpha(l(G))
    \end{vmatrix} 
\end{align*}
From Lemmas \ref{lemma::incident_BTB} and \ref{lemma::incident_BBT} we obtain

{\footnotesize
\begin{align*}
    P_{A_\alpha(T(G))}(\lambda) &= \begin{vmatrix}
    (\lambda - 2r\alpha)I_n -(1-\alpha)(BB^T-rI_n) & -(1-\alpha)B \\
    -(1-\alpha)B^T & (\lambda-2r\alpha) I_m -(1-\alpha)(B^TB - 2I_m)
    \end{vmatrix} \\
    &=\begin{vmatrix}
    (\lambda - 3r\alpha + r)I_n -(1-\alpha)BB^T & -(1-\alpha)B \\
    -(1-\alpha)B^T & (\lambda-2r\alpha + 2 - 2\alpha) I_m -(1-\alpha)B^TB
    \end{vmatrix}\\
    &=\begin{vmatrix}
    (\lambda +r(1-3\alpha))I_n -(1-\alpha)BB^T & -(1-\alpha)B \\
    -((1-\alpha)+\lambda + r(1-3\alpha))B^T + (1-\alpha)B^TBB^T & (\lambda-2\alpha(r+1) + 2) I_m
    \end{vmatrix}\\
    &=\begin{vmatrix}
    F_{11} & F_{12}\\
    F_{21} & F_{22}
    \end{vmatrix}
\end{align*}
}%
where $F_{11} = (\lambda +r(1-3\alpha))I_n -(1-\alpha)BB^T + \dfrac{(1-\alpha)B}{\lambda +2 -2\alpha(r+1)}(-((1-\alpha)+\lambda + r(1-3\alpha))B^T + (1-\alpha)B^TBB^T)$, $F_{12} =  \bigzero$, $F_{21}= -((1-\alpha)+\lambda + r(1-3\alpha))B^T + (1-\alpha)B^TBB^T$ and $F_{22} =  (\lambda-2\alpha(r+1) + 2) I_m$.

Then,

{\footnotesize
\begin{align}
   P_{A_\alpha(T(G))}(\lambda) &= (\lambda +2 -2\alpha(r+1))^{m} \Biggl \vert (\lambda + r(1 - 3\alpha))I_n - (1-\alpha)BB^T \nonumber \\
   &+ \dfrac{(1-\alpha)B}{\lambda +2 -2\alpha(r+1)}\Biggl(-((1-\alpha)+\lambda + r(1-3\alpha))B^T + (1-\alpha)B^TBB^T \Biggr) \Biggr\vert \nonumber \\
   &= (\lambda +2 -2\alpha(r+1))^{m-n} \Biggl \vert ((\lambda + r(1 - 3\alpha))I_n - (1-\alpha)BB^T)(\lambda +2 -2\alpha(r+1)) \nonumber \\
   &+ (1-\alpha)\Biggl(-((1-\alpha)+\lambda + r(1-3\alpha))BB^T + (1-\alpha)BB^TBB^T \Biggr) \Biggr\vert \nonumber \\
    &=(\lambda +2 -2\alpha(r+1))^{m-n} \Biggl \vert ((\lambda + r(1 - 3\alpha))I_n - (1-\alpha)BB^T)(\lambda +2 -2\alpha(r+1)) \nonumber \\
   &+ (1-\alpha)\Biggl(-((1-\alpha)+\lambda + r(1-3\alpha))I_n + (1-\alpha)BB^T \Biggr)BB^T \Biggr\vert \label{eq2::T(G)}
\end{align}
}%
Substituting  (\ref{eq::BBT_Aalpha}) in (\ref{eq2::T(G)}) we have
{\footnotesize
\begin{align*}
    P_{A_\alpha(T(G))}(\lambda) &= (\lambda +2 -2\alpha(r+1))^{m-n} \Biggl \vert (\lambda + 2 - 2\alpha(r+1))((\lambda - \alpha r)I_n - A_\alpha(G)) \\
    &+((-\lambda + \alpha r - 1 + \alpha)I_n + A_\alpha(G))((r - 2\alpha r)I_n +A_\alpha(G)) \Biggr \vert \\
    &= (\lambda +2 -2\alpha(r+1))^{m-n} \Biggl \vert A_\alpha(G)^2  + (\alpha(r+3) - 2\lambda + r -3)A_\alpha(G) \\
    & + (\lambda^2 +(-\alpha(r+2) -r +2)\lambda + \alpha r(r+1) -r)I_n \Biggr \vert
\end{align*}
}%
If $\sigma(A_\alpha(G)) = \left \{ \lambda_1(A_\alpha(G)), \ldots, \lambda_n(A_\alpha(G)) \right \}$, it follows that
{\footnotesize
\begin{align*}
    P_{A_\alpha(T(G))}(\lambda) &= (\lambda +2 -2\alpha(r+1))^{m-n} \prod_{i=1}^n \Biggl( \lambda_i(A_\alpha(G))^2  + (\alpha(r+3) - 2\lambda + r -3)\lambda_i(A_\alpha(G)) \\
    & + \lambda^2 +(-\alpha(r+2) -r +2)\lambda + \alpha r(r+1) -r \Biggr) \\
    &= (\lambda +2 -2\alpha(r+1))^{m-n} \prod_{i=1}^n \Biggl( \lambda^2 + (-\alpha(r+2) - 2\lambda_i(A_\alpha(G)) - r + 2) \lambda \\
    &+ \lambda_i(A_\alpha(G))^2 + (\alpha(r+3) + r -3) \lambda_i(A_\alpha(G)) + \alpha r(r+1) -r \Biggr)
\end{align*}
}%
whose roots are
{\footnotesize
\begin{equation*}
   \dfrac{1}{2}\Biggl(2(\alpha + \lambda_i(A_\alpha(G)) - 1) + r(\alpha + 1) \pm 
    \sqrt{(\alpha - 1)(\alpha(r+2)^2 -r^2 - 4(1+\lambda_i(A_\alpha(G))) )}\Biggr) 
\end{equation*}
}%
for $i = 1, \ldots, n$ and the result follows.
\end{proof}

\begin{theorem}\label{heo2::T(G)}
Let $G$ be a $r$-regular graph ($r > 1$) with $n$ vertices and $m$ edges. If the eigenvalues of $A(G)$ are $\lambda_1(A(G)), \ldots, \lambda_n(A(G))$ and $\alpha \in [0,1]$, then $A_\alpha(T(G))$ has $m-n$ eigenvalues equal to $2\alpha(r+1) - 2$ and the others $2n$ eigenvalues are
\begin{equation*} \label{eq3::T(G)}
    \dfrac{1}{2}\left( -2(\alpha - 1)(\lambda_i(A(G)) -1) + r(3\alpha + 1) \pm (\alpha - 1)\sqrt{4\lambda_i(A(G)) + r^2 + 4} \right)
\end{equation*}
for $i = 1, \ldots, n$.
\end{theorem}
\begin{proof}
From Lemma \ref{lemma::incident_BBT} and equation (\ref{eq2::T(G)}) we have that
{\footnotesize
\begin{multline*}
    P_{A_\alpha(T(G))}(\lambda) = (\lambda + 2 -2\alpha(r+1))^{m-n} \Biggl \vert (1-\alpha)^2A(G)^2 + (1-\alpha)(-3\alpha(r+1) +2\lambda -r+3)A(G)\\
    +(\lambda^2 -(r(3\alpha + 1) +2(\alpha - 1))\lambda + 2\alpha r^2(\alpha + 1) + r(3\alpha^2 -2\alpha -1))I_n \Biggr \vert
\end{multline*}
}%
If $\sigma(A(G)) = \left \{ \lambda_1(A(G)), \ldots, \lambda_n(A(G)) \right \}$, it follows that
{\footnotesize
\begin{align*}
     P_{A_\alpha(T(G))}(\lambda) &= (\lambda + 2 -2\alpha(r+1))^{m-n} \prod_{i=1}^n \Biggl((1-\alpha)^2\lambda_i(A(G))^2
     + (\alpha - 1)(-3\alpha(r+1) + 2\lambda \\
     &-r +3)\lambda_i(A(G)) + \lambda^2 - (r(3\alpha +1) +2(\alpha -1))\lambda +2\alpha r^2(\alpha+1) +r(3\alpha^2 -2\alpha -1) \Biggr) \\
     &=(\lambda + 2 -2\alpha(r+1))^{m-n} \prod_{i=1}^n \Biggl( \lambda^2 + ((\alpha-1)(2\lambda_i(A(G))-2) - r(3\alpha +1))\lambda \\
     &+ (1-\alpha)^2\lambda_i(A(G))^2
     - (\alpha -1)(3r\alpha +3\alpha + r-3)\lambda_i(A(G)) + 2\alpha r^2(\alpha + 1)\\
     & + r(3\alpha^2 -2\alpha -1)\Biggr)
\end{align*}
}%
whose roots are,
\begin{equation*}
    \dfrac{1}{2}\left(-2(\alpha -1)(\lambda_i(A(G)) -1) + r(3\alpha +1) \pm (\alpha -1)\sqrt{4\lambda_i(A(G)) + r^2 +4} \right)
\end{equation*}
for $i = 1, \ldots, n$ and the result follows.
\end{proof}

\subsection*{Acknowledgments}

The research of C. S. Oliveira is supported by CNPq Grant 304548/2020-0.

\bibliographystyle{unsrt}  
\bibliography{charpol}  

\begin{thebibliography}{10}

\bibitem{VN17}
V.~Nikiforov.
\newblock {Merging the A- and Q- spectral theories}.
\newblock {\em {Applicable Analysis and Discrete Mathematics}},
  {11}({1}):81--107, 2017.

\bibitem{BRONDANI2020}
A.E. Brondani, F.A.M. França, and C.S. Oliveira.
\newblock Positive semidefiniteness of {A}$_\alpha$({G}) on some families of
  graphs.
\newblock {\em Discrete Applied Mathematics}, 2020.

\bibitem{NIKIFOROV2017156}
Vladimir Nikiforov and Oscar Rojo.
\newblock A note on the positive semidefiniteness of {A}$_\alpha$.
\newblock {\em Linear Algebra and its Applications}, 519:156--163, 2017.

\bibitem{BRONDANI2019209}
André~Ebling Brondani, Carla~Silva Oliveira, Francisca Andrea~Macedo França,
  and Leonardo {de Lima}.
\newblock {A}$_\alpha$-spectrum of a firefly graph.
\newblock {\em Electronic Notes in Theoretical Computer Science}, 346:209--219,
  2019.

\bibitem{MUHAMMAD2020}
Muhammad~Ateeq Tahir and Xiao-Dong Zhang.
\newblock Coronae graphs and their $\alpha$-eigenvalues.
\newblock {\em Bulletin of the Malaysian Mathematical Sciences Society},
  43:2911–2927, 2020.

\bibitem{Li2019TheS}
Shuchao Li and Shujing Wang.
\newblock The {A}$_\alpha$- spectrum of graph product.
\newblock {\em The Electronic Journal of Linear Algebra}, 2019.

\bibitem{CHEN2019343}
Yuanyuan Chen, Dan Li, and Jixiang Meng.
\newblock On the second largest {A}$_\alpha$({G})-eigenvalues of graphs.
\newblock {\em Linear Algebra and its Applications}, 580:343--358, 2019.

\bibitem{LIN2018210}
Huiqiu Lin, Jie Xue, and Jinlong Shu.
\newblock On the {A}$_\alpha$-spectra of graphs.
\newblock {\em Linear Algebra and its Applications}, 556:210--219, 2018.

\bibitem{Lin2017GraphsDB}
Huiqiu Lin, Xiaogang Liu, and Jie Xue.
\newblock Graphs determined by their {A}$_\alpha$-spectra.
\newblock {\em arXiv: Combinatorics}, 2017.

\bibitem{Tahir2018}
Muhammad~Ateeq Tahir and Xiao-Dong Zhang.
\newblock Graphs with three distinct $\alpha$-eigenvalues.
\newblock {\em Acta Mathematica Vietnamica}, 43(4):649--659, 6 2018.

\bibitem{LIU2020111917}
Shuting Liu, Kinkar~Chandra Das, and Jinlong Shu.
\newblock On the eigenvalues of {A}$_\alpha$-matrix of graphs.
\newblock {\em Discrete Mathematics}, 343(8):111917, 2020.

\bibitem{LIN2018430}
Huiqiu Lin, Xing Huang, and Jie Xue.
\newblock A note on the {A}$_\alpha$-spectral radius of graphs.
\newblock {\em Linear Algebra and its Applications}, 557:430--437, 2018.

\bibitem{WANG2020210}
Sai Wang, Dein Wong, and Fenglei Tian.
\newblock Bounds for the largest and the smallest {A}$_\alpha$ eigenvalues of a
  graph in terms of vertex degrees.
\newblock {\em Linear Algebra and its Applications}, 590:210--223, 2020.

\bibitem{LIU2020347}
Shuting Liu, Kinkar~Chandra Das, Shaowei Sun, and Jinlong Shu.
\newblock On the least eigenvalue of {A}$_\alpha$-matrix of graphs.
\newblock {\em Linear Algebra and its Applications}, 586:347--376, 2020.

\bibitem{SP2022}
Shariefuddin Pirzada.
\newblock Two upper bounds on the {A}$_\alpha$-spectral radius of a connected
  graph.
\newblock {\em Communications in Combinatorics and Optimization}, 7(1):53--57,
  2022.

\bibitem{LIU2018274}
Xiaogang Liu and Shunyi Liu.
\newblock On the {A}$_\alpha$-characteristic polynomial of a graph.
\newblock {\em Linear Algebra and its Applications}, 546:274--288, 2018.

\bibitem{ADB_Equitable}
Amanda Francis, Dallas Smith, and Benjamin Webb.
\newblock General equitable decompositions for graphs with symmetries, 2018.

\bibitem{horn2013matrix}
R.A. Horn and C.R. Johnson.
\newblock {\em Matrix Analysis}.
\newblock Matrix Analysis. Cambridge University Press, New York, 2013.

\bibitem{silvester_2000}
John~R. Silvester.
\newblock Determinants of block matrices.
\newblock {\em The Mathematical Gazette}, 84(501):460–467, 2000.

\bibitem{cvetkovic2009introduction}
D.~Cvetkovi{\'c}, P.~Rowlinson, and S.~Simi{\'c}.
\newblock {\em An Introduction to the Theory of Graph Spectra}.
\newblock London Mathematical Society Student Texts. Cambridge University
  Press, Cambridge, 2010.

\bibitem{Goldberg59}
K.~Goldberg.
\newblock Principal sub-matrices of a full-rowed non-negative matrix.
\newblock {\em Journal of Research of the National Bureau of Standards, Section
  B: Mathematics and Mathematical Physics}, 63B(1):19--20, 1959.

\end{thebibliography}

\end{document}